\documentclass[12pt,reqno]{amsart}%
\usepackage{amsmath, amsfonts, amssymb, amsthm, amscd, amsbsy}
\usepackage[usenames,dvipsnames,svgnames,x11names,hyperref]{xcolor}
\usepackage{geometry}
\usepackage{enumerate}
\usepackage{graphicx}
\usepackage[pagebackref]{hyperref}
\usepackage{amsmath}
\usepackage{amsfonts}
\usepackage{amssymb}%
\providecommand{\U}[1]{\protect\rule{.1in}{.1in}}
\hypersetup{colorlinks=true,
breaklinks=true,
urlcolor=NavyBlue,
linkcolor=Fuchsia,
bookmarksopen=false,
filecolor=black,
citecolor=ForestGreen,
linkbordercolor=red
}
\newtheorem{theorem}{Theorem}[section]
\theoremstyle{plain}

\newtheorem{corollary}{Corollary}[section]

\newtheorem{lemma}{Lemma}[section]

\newtheorem{remark}{Remark}[section]

\numberwithin{equation}{section}
\allowdisplaybreaks
\AtBeginDocument{\hypersetup{pdfborder={0 0 0.01}}}
\newcommand{\rr}{\mathbb{R}}
\newcommand{\dx}{\mathrm{dx}}

\begin{document}
\title{Sharp second order uncertainty principles}
\author{Cristian Cazacu}
\address{Cristian Cazacu: $^1$Faculty of Mathematics and Computer Science \\
University of Bucha-rest\\
010014 Bucharest, Romania\\
\&
$^2$Gheorghe Mihoc-Caius Iacob Institute of Mathematical\\
Statistics and Applied Mathematics of the Romanian Academy\\
050711 Bucharest, Romania\\
\&
$^3$Member in the research grant no. PN-III-P1-1.1-TE-2019-0456\\ University Politehnica of Bucharest\\ 
060042, Bucharest, Romania
}
\email{cristian.cazacu@fmi.unibuc.ro}
\author{Joshua Flynn}
\address{Joshua Flynn: Department of Mathematics\\
University of Connecticut\\
Storrs, CT 06269, USA}
\email{joshua.flynn@uconn.edu}
\author{Nguyen Lam}
\address{Nguyen Lam: School of Science \& Environment\\
Grenfell Campus, Memorial University of Newfoundland\\
Corner Brook, NL A2H5G4, Canada }
\email{nlam@grenfell.mun.ca}
\thanks{C.C. was partially supported by CNCS-UEFISCDI Romania, Grant No. PN-III-P1-1.1-TE-2019-0456. J.F. was partially supported by a Simons Collaboration grant from the Simons Foundation.}
\date{\today}

\begin{abstract}
We study sharp second order inequalities of Caffarelli-Kohn-Nirenberg type in the euclidian
space $\mathbb{R}^{N}$, where $N$ denotes the dimension. This analysis is equivalent to the study of uncertainty principles for special classes of vector fields.  In particular, we
show that when switching from scalar fields $u: \rr^n\rightarrow \mathbb{C}$ to vector fields of the form $\vec{u}:=\nabla U$ ($U$ being a scalar field) the best constant in the Heisenberg Uncertainty Principle (HUP) increases
from $\frac{N^{2}}{4}$ to $\frac{(N+2)^{2}}{4}$, and the optimal constant in
the Hydrogen Uncertainty Principle (HyUP) improves from $\frac{\left(  N-1\right)
^{2}}{4}$ to $\frac{(N+1)^{2}}{4}$. As a consequence of our results we   
answer to the open question of Maz'ya \cite[Section 3.9]{M} in the case $N=2$
regarding the HUP for divergence free vector fields.

\end{abstract}
\maketitle

\section{Introduction}

In quantum mechanics, the uncertainty principle implies that the position and the
momentum of an object cannot both be measured exactly, at the same time, not even
in theory. It is one of the most important differences between quantum and
classical mechanics. The most well-known formulation of the uncertainty principle is probably the
Heisenberg-Pauli-Weyl Uncertainty Principle (see, e.g., \cite{MR2583992, MR3363447}) (henceforth, HUP for short), which can be stated as the following inequality 
\begin{equation}
\int_{\mathbb{R}^{N}}|\nabla u|^{2}\mathrm{dx}\int_{\mathbb{R}^{N}}%
|x|^{2}|u|^{2}\mathrm{dx}\geq\mu^{\star}(N)\left(  \int_{\mathbb{R}^{N}%
}|u|^{2}\mathrm{dx}\right)  ^{2}, \quad \mu^{\star}(N):=\frac{N^{2}}{4},\label{HUP}%
\end{equation}
which holds for any function $u\in C_{c}^{\infty}(\mathbb{R}^{N})$ (i.e, the space of smooth compactly supported functions) but it can also be extended to functions $u$ in the Schwartz space $\mathcal{S}(\mathbb{R}^{N})$ or in appropriate Sobolev spaces. 

 The physical meaning of \eqref{HUP} asserts that 
if $u$ is normalized to $\left\lVert u \right\rVert_{2}=1$, since $p=-i\nabla$ denotes the momentum operator, then the position $\left\lVert x u \right\rVert_{2}$ and the momentum $\left\lVert \nabla u \right\rVert_{2}=\left\lVert p u \right\rVert_{2}$  cannot be small enough simultaneously because of the lower bound $\left\lVert pu \right\rVert_{2} \left\lVert xu \right\rVert_{2}\geq\frac{N}{2}$.
Also, using the Plancherel theorem, \eqref{HUP} implies that a function $u$ and its Fourier transform $\hat{u}$ may not be sharply localized at the origin simultaneously (i.e., a function is sharply localized/concentrated at the origin  if its support can be contained in a small neighborhood of the origin) since (assuming $\left\lVert  u \right\rVert_{2}=1$)
$$\int_{\rr^N} 4\pi^2 |\xi|^2|\hat{u}(\xi)|^2 d\xi \int_{\rr^N} |x|^2 |u(x)|^2 \mathrm{dx} \geq \mu^\star(N), \quad \hat{u}(\xi):=\int_{\rr^N} u(x)e^{-2\pi i x\cdot \xi } \dx. $$
  It is well-known that the constant $\mu^{\star}(N)$ is optimal, i.e. the largest
constant which validates inequality \eqref{HUP} (see, e.g., \cite{FS97}). Moreover, equality in \eqref{HUP} is not attained in the space
$C_{c}^{\infty}(\mathbb{R}^{N})$, but in a slightly larger space which is the
Schwartz space $\mathcal{S}(\mathbb{R}^{N})$ and the minimizers are the
Gaussian profiles of the form $u(x)=\alpha e^{-\beta|x|^{2}}$, $\beta>0$ (see
Section \ref{main_results} for further details). In fact the proof of
\eqref{HUP} is straightforward by applying integration by parts combined with
the Cauchy-Buniakovski-Schwarz inequality:
\begin{align*}
\int_{\mathbb{R}^{N}}|u|^{2}\mathrm{dx} &  =\frac{1}{N}\int_{\mathbb{R}^{N}%
}\mathrm{\operatorname{div}}(x)|u|^{2}\mathrm{dx}=-\frac{1}{N}\int
_{\mathbb{R}^{N}}x\cdot\nabla(|u|^{2})\mathrm{dx}=-\frac{2}{N}\operatorname{Re}\int
_{\mathbb{R}^{N}}(x\cdot\nabla u)\overline{u}\mathrm{dx}\\
&  \leq\frac{2}{N}\int_{\mathbb{R}^{N}}|x||\nabla u||u|\mathrm{dx}\leq\frac
{2}{N}\left(  \int_{\mathbb{R}^{N}}|x|^{2}|u|^{2}\mathrm{dx}\right)
^{1/2}\left(  \int_{\mathbb{R}^{N}}|\nabla u|^{2}\mathrm{dx}\right)  ^{1/2}.
\end{align*}
The HUP belongs to larger family of interpolation inequalities called the Caffarelli-Kohn-Nirenberg (CKN) inequalities introduced in \cite{CKN} to study the Navier-Stokes equation and the regularity of particular solutions \cite{CKN2}.
The family of CKN inequalities includes as special cases, among others, the HUP and Hardy's inequality 
\begin{equation}
  \int_{\mathbb{R}^{N}}|\nabla{u}|^{2}\dx\geq\left( \frac{N-2}{2} \right)^{2} \int_{\mathbb{R}^{N}}\frac{|u|^{2}}{|x|^{2}}\dx .
  \label{eq:hardy-uncertainty-principle}
\end{equation}
Observe that \eqref{eq:hardy-uncertainty-principle} is also an uncertainty principle in that localization in $u$ at the origin implies its gradient norm $\left\lVert \nabla u \right\rVert_{2}$ must be large.
 Important contributions concerning sharp constants and their minimizers for CKN inequalities have been highlighted afterwards, see, e.g. \cite{CW}, \cite{CC}, \cite{Co}.  
Recently, in connection with our purpose, we point out some work by Hamamoto et. al. (see \cite{H1}, \cite{H2} and the references therein)  which has been done in the topic of Hardy inequality for vector fields.

Related to the HUP and Hardy's inequality is the so-called Hydrogen Uncertainty Principle (HyUP) that can be stated as follows: for any $u\in C_{c}^{\infty}(\mathbb{R}%
^{N})$, there holds
\begin{equation}
\int_{\mathbb{R}^{N}}|\nabla u|^{2}\mathrm{dx}\int_{\mathbb{R}^{N}}%
|u|^{2}\mathrm{dx}\geq\nu^{\star}(N)\left(  \int_{\mathbb{R}^{N}}\frac
{|u|^{2}}{\left\vert x\right\vert }\mathrm{dx}\right)  ^{2}, \quad \nu^{\star}(N):=\frac{\left(  N-1\right)  ^{2}}{4}.\label{HyUP}%
\end{equation}
The HyUP is related to the ground state of a single electron and single fixed nucleus system (i.e., a hydrogenic atom) with Coulomb interaction.
Indeed, the singular weight $|x|^{-1}$ arises from the Coulomb potential, and, in fact, using \eqref{eq:hardy-uncertainty-principle}, one may show the quantum mechanical energy of the hydrogen atom is finite.
The constant $\nu^{\star}(N)$ in \eqref{HyUP} is also
optimal and the minimizers are of the form $u(x)=\alpha e^{-\beta|x|}$,
$\beta>0$ (see, e.g. \cite{Fra11}). Notice that in this case the minimizers are not in $\mathcal{S}(\rr^N)$ but in a Sobolev space, namely, for our purpose, $W^{2,2}(\rr^N)$.

Uncertainty principles such as HUP and HyUP have several physical and mathematical applications.
As mentioned above, in physics, uncertainty principles may be used for establishing stability of matter.
For example, the HyUP may be used to show the stability of a hydrogenic atom in a magnetic field (e.g., see \cite{MR836003,MR2583992}).
For more general systems (e.g., a many-electron atom or many fermion systems), stronger uncertainty principles are required and used for establishing stability (see \cite{MR0456083}).
In mathematics, uncertainty principles may be used to establish rigidity results (e.g., \cite{MR3862150}), to study variable coefficient partial differential operators (e.g., \cite{MR707957}) such as certain Schr\"odinger operators, and so on.

\medskip

According to Maz'ya \cite[Section 3.9]{M} a related open problem to this
subject concerns finding the best constant in \eqref{HUP} when we replace $u$
by a divergence-free vector field $\vec{U}$. Namely, determine the best
constant say, denoted by $\mu^{\star\ast}(N)$, in the inequality
\begin{equation}
\int_{\mathbb{R}^{N}}|\nabla\vec{U}|^{2}\mathrm{dx}\int_{\mathbb{R}^{N}%
}|x|^{2}|\vec{U}|^{2}\mathrm{dx}\geq\mu^{\star\ast}(N)\left(  \int
_{\mathbb{R}^{N}}|\vec{U}|^{2}\mathrm{dx}\right)  ^{2},\quad\forall\vec{U}\in
C_{c}^{\infty}(\mathbb{R}^{N}),\quad\mathrm{\operatorname{div}}\vec
{U}=0.\label{HUP_div}%
\end{equation}
Notice that $\mu^{\star\star}(N)\geq\mu^{\star}(N)=N^{2}/4$. Indeed, by
applying \eqref{HUP} for each component $U_{i}$, $i=1,...,N$ and in view of
the Cauchy-Buniakovski-Schwarz inequality we successively obtain
\begin{align*}
\int_{\mathbb{R}^{N}}|\nabla\vec{U}|^{2}\mathrm{dx}\int_{\mathbb{R}^{N}%
}|x|^{2}|\vec{U}|^{2}\mathrm{dx} &  =\left(  \sum_{i=1}^{N}\int_{\mathbb{R}%
^{N}}|\nabla U_{i}|^{2}\mathrm{dx}\right)  \left(  \sum_{i=1}^{N}%
\int_{\mathbb{R}^{N}}|x|^{2}|U_{i}|^{2}\mathrm{dx}\right)  \\
&  \geq\left(  \sum_{i=1}^{N}\sqrt{\int_{\mathbb{R}^{N}}|\nabla U_{i}%
|^{2}\mathrm{dx}\int_{\mathbb{R}^{N}}|x|^{2}|U_{i}|^{2}\mathrm{dx}}\right)
^{2}\\
&  \geq\left(  \frac{N}{2}\int_{\mathbb{R}^{N}}|U_{i}|^{2}\mathrm{dx}\right)
^{2}=\frac{N^{2}}{4}\left(  \int_{\mathbb{R}^{N}}|\vec{U}|^{2}\mathrm{dx}%
\right)  ^{2}.
\end{align*}
Similarly, if we replace $u$ by a divergence-free vector field $\vec{U}$ then
$\nu^{\star\star}(N)\geq\frac{\left(  N-1\right)  ^{2}}{4}$. Here $\nu
^{\star\star}(N)$ is the best constant of the following inequality
\begin{equation}
\int_{\mathbb{R}^{N}}|\nabla\vec{U}|^{2}\mathrm{dx}\int_{\mathbb{R}^{N}}%
|\vec{U}|^{2}\mathrm{dx}\geq\nu^{\star\ast}(N)\left(  \int_{\mathbb{R}^{N}%
}\frac{|\vec{U}|^{2}}{\left\vert x\right\vert }\mathrm{dx}\right)  ^{2}%
,\quad\forall\vec{U}\in C_{c}^{\infty}(\mathbb{R}^{N}),\quad
\mathrm{\operatorname{div}}\vec{U}=0.\label{HyUP_div}%
\end{equation}

\paragraph*{\textbf{The case $N=2$}}

It is well-known that in this case a divergence-free vector can be written in
the form $\vec{U}=(-u_{x_{2}}, u_{x_{1}})$ where $u$ is a scalar field. If
$\vec{U}\in C_{c}^{\infty}(\mathbb{R}^{2})$ then also $u\in C_{c}^{\infty
}(\mathbb{R}^{2})$.

Then, after integration by parts we get%

\[
\int_{\mathbb{R}^{2}}|\nabla\vec{U}|^{2}\mathrm{dx}=\int_{\mathbb{R}^{2}%
}|u_{x_{1}x_{1}}|^{2}+|u_{x_{2}x_{2}}|^{2}+2|u_{x_{1}x_{2}}|^{2}%
\mathrm{dx}=\int_{\mathbb{R}^{2}}|\Delta u|^{2}\mathrm{dx}.
\]
Therefore \eqref{HUP_div} is equivalent to

\begin{equation}
\int_{\mathbb{R}^{2}}|\Delta u|^{2}\mathrm{dx}\int_{\mathbb{R}^{2}}%
|x|^{2}|\nabla u|^{2}\mathrm{dx}\geq\mu^{\star\ast}(2)\left(  \int
_{\mathbb{R}^{2}}|\nabla u|^{2}\mathrm{dx}\right)  ^{2},
\label{HUP_div_scalar}%
\end{equation}
and \eqref{HyUP_div} is equivalent to%

\begin{equation}
\int_{\mathbb{R}^{2}}|\Delta u|^{2}\mathrm{dx}\int_{\mathbb{R}^{2}}|\nabla
u|^{2}\mathrm{dx}\geq\nu^{\star\ast}(2)\left(  \int_{\mathbb{R}^{2}}%
\frac{|\nabla u|^{2}}{\left\vert x\right\vert }\mathrm{dx}\right)
^{2}.\label{HyUP_div_scalar}%
\end{equation}

\medskip

The primary goal of this article is to investigate several \emph{
second order uncertainty principles} of the forms (\ref{HUP_div_scalar}%
) and (\ref{HyUP_div_scalar}) providing their sharp constants and
optimizers in any dimension $N\geq1$. Particularly, we determine $\mu
^{\star\star}(2)$ in \eqref{HUP_div} or \eqref{HUP_div_scalar} and answer to the
question of Maz'ya in the case $N=2$.

\section{Main results}

\label{main_results}

Let $\mathcal{S}(\mathbb{R}^{N})$ be the Schwartz space of smooth functions, i.e., those which decay faster than any polynomial at infinity, namely
\[
\mathcal{S}(\mathbb{R}^{N}):=\{u\in C^{\infty}(\mathbb{R}^{N}) \ |
\ \sup_{x\in\mathbb{R}^{N}} |x^{\alpha}D^{\beta}u(x)|<\infty, \quad
\forall\alpha, \beta\text{ multi-indices}\},
\]
where the multi-indices $\alpha$ and $\beta$ are defined by
\[
\alpha=(\alpha_{1}, \ldots, \alpha_{N}), \quad\beta=(\beta_{1}, \ldots,
\beta_{N}), \quad\text{ with }\alpha_{i}, \beta_{i} \in\mathbb{N}, \forall
i=1,\ldots, N.
\]
For $x=(x_{1}, \ldots, x_{N})$, the monomial $x^{\alpha}$ and the partial
derivatives $D^{\beta}u(x)$ are given by
\[
x^{\alpha}:= x_{1}^{\alpha_{1}}\ldots x_{N}^{\alpha_{N}}, \quad D^{\beta
}u(x):=\frac{\partial^{\beta_{1}+\ldots+\beta_{N}}}{\partial x_{1}^{\beta_{1}%
}\ldots x_{N}^{\beta_{N}}}u(x).
\]

In harmonic analysis and PDEs, the Schwartz space is famous for its invariance property under the Fourier transform.
They are also important to state our main results.
In fact, we point out that any functional inequality under consideration in this paper which is
valid for compactly supported functions can be easily extended by density to
functions in the Schwartz space $\mathcal{S}(\mathbb{R}^{N})$. Thus, it makes
sense to study such inequalities in $\mathcal{S}(\mathbb{R}^{N})$.

In connection with inequalities \eqref{HUP_div_scalar}-\eqref{HyUP_div_scalar} we have the following
main results.

\begin{theorem}
\label{th1} Let $N\geq1$ and $u\in\mathcal{S}(\mathbb{R}^{N})$. Then the
following inequality holds
\begin{equation}
\int_{\mathbb{R}^{N}}|\Delta u|^{2}\mathrm{dx}\int_{\mathbb{R}^{N}}%
|x|^{2}|\nabla u|^{2}\mathrm{dx}\geq\frac{(N+2)^{2}}{4}\left(  \int
_{\mathbb{R}^{N}}|\nabla u|^{2}\mathrm{dx}\right)  ^{2},\label{main1}%
\end{equation}
where the constant $\frac{(N+2)^{2}}{4}$ is optimal and it is attained for
Gaussian profiles of the form $u(x)=\alpha e^{-\beta|x|^{2}}$, $\beta>0$,
$\alpha\in\mathbb{%
%TCIMACRO{\U{2102} }%
%BeginExpansion
\mathbb{C}
%EndExpansion
}$.
\end{theorem}

\begin{remark}
Theorem \ref{th1} gives a response to the Maz'ya question in the case $N=2$
since the best constant in \eqref{HUP_div} increases from $1$ to $4$ with
respect to the best constant in \eqref{HUP}. Also, this improves the second order inequalities in \cite{CC}. 
\end{remark}

Now, we denote $\mathcal{R}_{1}:=\frac{x}{|x|}\cdot\nabla$ and $\mathcal{R}%
_{2}:=\mathcal{R}_{1}^{2}+\frac{N-1}{r}\mathcal{R}_{1}$ where $\mathcal{R}%
_{1}^{2}:=\mathcal{R}_{1}\circ\mathcal{R}_{1}$, which represent the radial
derivative, and the radial Laplacian, respectively. We claim

\begin{theorem}
\label{th2}Let $N\geq1$ and $u\in\mathcal{S}(\mathbb{R}^{N})$. Then the
following inequality holds
\begin{equation}
\int_{\mathbb{R}^{N}}|\mathcal{R}_{2}u|^{2}\mathrm{dx}\int_{\mathbb{R}^{N}%
}|x|^{2}|\mathcal{R}_{1}u|^{2}\mathrm{dx}\geq\frac{(N+2)^{2}}{4}\left(
\int_{\mathbb{R}^{N}}|\mathcal{R}_{1}u|^{2}\mathrm{dx}\right)  ^{2}%
.\label{main2}%
\end{equation}
Moreover, \eqref{main2} is optimal with the best constant $\frac{(N+2)^{2}}%
{4}$, which is achieved for Gaussian profiles of the form $u(x)=\alpha
e^{-\beta|x|^{2}}$, $\beta>0$, $\alpha\in\mathbb{%
%TCIMACRO{\U{2102} }%
%BeginExpansion
\mathbb{C}
%EndExpansion
}$.
\end{theorem}

\begin{remark}
\label{r2.2}The spherical harmonics decomposition implemented in the proof of
Theorem \ref{th1} and Theorem \ref{th2} (more specifically, in Lemmas \ref{l3.1}-\ref{l3.2}) shows that
\[
\int_{\mathbb{R}^{N}}|\Delta u|^{2}\mathrm{dx}\geq\int_{\mathbb{R}^{N}%
}|\mathcal{R}_{2}u|^{2}\mathrm{dx},\quad\forall u\in \mathcal{S}%
(\mathbb{R}^{N}),
\]
unless $N=2$. Thus, when $N\neq2$, in view of \eqref{main2} we also get
\begin{equation}
\int_{\mathbb{R}^{N}}|\Delta u|^{2}\mathrm{dx}\int_{\mathbb{R}^{N}}%
|x|^{2}|\mathcal{R}_{1}u|^{2}\mathrm{dx}\geq\frac{(N+2)^{2}}{4}\left(
\int_{\mathbb{R}^{N}}|\mathcal{R}_{1}u|^{2}\mathrm{dx}\right)  ^{2}%
.\label{main3}%
\end{equation}
Also, \eqref{main3} is optimal with the best constant $\frac{(N+2)^{2}}{4}$,
which is achieved for Gaussian profiles of the form $u(x)=\alpha
e^{-\beta|x|^{2}}$, $\beta>0$, $\alpha\in\mathbb{%
%TCIMACRO{\U{2102} }%
%BeginExpansion
\mathbb{C}
%EndExpansion
}$.
\end{remark}

\begin{remark}
As a consequence of Theorems \ref{th1}-\ref{th2}, since $\left\vert \nabla
u\right\vert \geq\left\vert \mathcal{R}_{1}u\right\vert $ pointwise, we
trivially obtain the following optimal inequalities%
\begin{equation}
\int_{\mathbb{R}^{N}}|\mathcal{R}_{2}u|^{2}\mathrm{dx}\int_{\mathbb{R}^{N}%
}|x|^{2}|\nabla u|^{2}\mathrm{dx}\geq\frac{(N+2)^{2}}{4}\left(  \int
_{\mathbb{R}^{N}}|\mathcal{R}_{1}u|^{2}\mathrm{dx}\right)  ^{2}%
,\label{sec_main2}%
\end{equation}%
\begin{equation}
\int_{\mathbb{R}^{N}}|\Delta u|^{2}\mathrm{dx}\int_{\mathbb{R}^{N}}%
|x|^{2}|\nabla u|^{2}\mathrm{dx}\geq\frac{(N+2)^{2}}{4}\left(  \int
_{\mathbb{R}^{N}}|\mathcal{R}_{1}u|^{2}\mathrm{dx}\right)  ^{2},\text{ }%
N\neq2\text{.}\label{sec_main4}%
\end{equation}
Moreover, inequalities \eqref{sec_main2}-\eqref{sec_main4} are optimal with
the best constant $\frac{(N+2)^{2}}{4}$, which is achieved for Gaussian
profiles of the form $u(x)=\alpha e^{-\beta|x|^{2}}$, $\beta>0$, $\alpha
\in\mathbb{%
%TCIMACRO{\U{2102} }%
%BeginExpansion
\mathbb{C}
%EndExpansion
}$.
\end{remark}

Concerning HyUP, we will prove that

\begin{theorem}
\label{th3}Let $N\geq5$ and $u\in W^{2,2}\left(  \mathbb{R}^{N}\right)  $.
Then the following inequality holds
\begin{equation}
\int_{\mathbb{R}^{N}}|\Delta u|^{2}\mathrm{dx}\int_{\mathbb{R}^{N}}|\nabla
u|^{2}\mathrm{dx}\geq\frac{(N+1)^{2}}{4}\left(  \int_{\mathbb{R}^{N}}%
\frac{|\nabla u|^{2}}{\left\vert x\right\vert }\mathrm{dx}\right)
^{2}\label{main4}%
\end{equation}
where the constant $\frac{(N+1)^{2}}{4}$ is optimal and it is attained for
functions of the form $u(x)=\alpha\left(  1+\beta\left\vert x\right\vert
\right)  e^{-\beta|x|}$, $\beta>0$, $\alpha\in\mathbb{%
%TCIMACRO{\U{2102} }%
%BeginExpansion
\mathbb{C}
%EndExpansion
}$.
\end{theorem}

\begin{remark}
We conjecture that (\ref{main4}) still holds for $2\leq N\leq4$.
\end{remark}

\begin{theorem}
\label{th4}Let $N\geq2$ and $u\in W^{2,2}\left(  \mathbb{R}^{N}\right)  $.
Then the following inequality holds
\begin{equation}
\int_{\mathbb{R}^{N}}|\mathcal{R}_{2}u|^{2}\mathrm{dx}\int_{\mathbb{R}^{N}%
}|\mathcal{R}_{1}u|^{2}\mathrm{dx}\geq\frac{(N+1)^{2}}{4}\left(
\int_{\mathbb{R}^{N}}\frac{|\mathcal{R}_{1}u|^{2}}{\left\vert x\right\vert
}\mathrm{dx}\right)  ^{2}.\label{main5}%
\end{equation}
Also, \eqref{main5} is optimal with the best constant $\frac{(N+1)^{2}}{4}$,
which is achieved for functions of the form $u(x)=\alpha\left(  1+\beta
\left\vert x\right\vert \right)  e^{-\beta|x|}$, $\beta>0$, $\alpha\in\mathbb{%
%TCIMACRO{\U{2102} }%
%BeginExpansion
\mathbb{C}
%EndExpansion
}$.
\end{theorem}

\medskip

\paragraph{\bf HUP and HyUP for 1-forms.} Our main results in Theorems \ref{th1}-\ref{th3} could be fashionably stated in terms of differential forms. 
For instance, inequality \eqref{main1} is equivalent to inequality
\eqref{HUP_div} when switching from divergence-free vector fields to
generalized "rotational-free" vector fields. This can be easily expressed in
terms of differential forms as follows.

\paragraph*{\textbf{Exact/closed 1-forms}}

First note that a vector field $\vec{U}:\mathbb{R}^{N}\rightarrow\mathbb{C}^{N}$ can be
interpreted as a 1-form. We say that $\vec{U}$ is an \emph{exact
1-form} if there exists a scalar function $u: \mathbb{R}^{N}\rightarrow
\mathbb{C}$ (0-form) such that $\vec{U}=\nabla u$. In particular, $\vec{U}$
is \emph{a closed form}, i.e. $\text{d}\vec{U}=0$, where $\mathrm{d}$ is the
exterior derivative. The space $\mathbb{R}^{N}$ is contractible to a point and
by the Poincar\'{e} lemma (see, e.g. \cite[Cor. 18, pp. 225]{S}) it follows
that $\vec{U}$ is exact form if and only $\vec{U}$ is a closed form. More
detailed necessary information on the formalism of differential forms can be
extracted from \cite{CK}.

Hence, as a consequence of Theorem \ref{th1} we have

\begin{corollary}
\label{cor1} Assume that $\vec{U}$ is a smooth closed 1-form and compactly
supported. Then the following inequality holds:
\begin{equation}
\int_{\mathbb{R}^{N}}|\nabla\vec{U}|^{2}\mathrm{dx}\int_{\mathbb{R}^{N}%
}|x|^{2}|\vec{U}|^{2}\mathrm{dx}\geq\frac{(N+2)^{2}}{4}\left(  \int
_{\mathbb{R}^{N}}|\vec{U}|^{2}\mathrm{dx}\right)  ^{2}. \label{ineq2}%
\end{equation}

\end{corollary}

Indeed, since $\vec{U}$ is also exact 1-form there exists a scalar function
$u\in C_{c}^{\infty}(\mathbb{R}^{N})$ such that $\vec{U}=\nabla u$. Computing
after integrations by parts we have
\[
\int_{\mathbb{R}^{N}}|\nabla\vec{U}|^{2}\mathrm{dx}=\sum_{i=1}^{N}%
\int_{\mathbb{R}^{N}}|\nabla u_{x_{i}}|^{2}\mathrm{dx}=\sum_{i,j=1}^{N}%
\int_{\mathbb{R}^{N}}|u_{x_{i}x_{j}}|^{2}\mathrm{dx}=\int_{\mathbb{R}^{N}%
}|\Delta u|^{2}\mathrm{dx}.
\]
Then \eqref{main1} reduces to \eqref{ineq2} and Corollary \ref{cor1} is
proven. \newline

As a consequence of Corollary \ref{cor1} we obtain

\begin{corollary}
[rotational-free vector fields]\label{cor2} Assume that $N=3$ and $\vec{U}$ is
a smooth compactly supported vector field such that $\emph{curl}\vec{U}=0$.
Then
\[
\int_{\mathbb{R}^{3}}|\nabla\vec{U}|^{2}\mathrm{dx}\int_{\mathbb{R}^{3}%
}|x|^{2}|\vec{U}|^{2}\mathrm{dx}\geq\frac{25}{4}\left(  \int_{\mathbb{R}^{3}%
}|\vec{U}|^{2}\mathrm{dx}\right)  ^{2}.
\]

\end{corollary}

Indeed, in the case $N=3$ it holds $\text{d}\vec{U}=\text{curl}\vec{U}=0$ i.e.
$\vec{U}$ is closed.\newline

Similarly, we can rewrite Theorem \ref{th3} as 
\begin{corollary}
	\label{cor3} Assume that $N\geq 5$ and  $\vec{U}$ is a smooth closed 1-form and compactly
	supported. Then the following inequality holds
	\begin{equation}
	\int_{\mathbb{R}^{N}}|\nabla\vec{U}|^{2}\mathrm{dx}\int_{\mathbb{R}^{N}%
	}|\vec{U}|^{2}\mathrm{dx}\geq\frac{(N+1)^{2}}{4}\left(  \int
	_{\mathbb{R}^{N}}\frac{|\vec{U}|^{2}}{|x|}\mathrm{dx}\right)  ^{2}. \label{ineq2-josh}%
	\end{equation}
	
\end{corollary}

\section{Some useful computations}

\label{sec_useful}

Let $N\geq2$. We implement the well-known idea of decomposing $u$ into
spherical harmonics as follows. We apply the coordinates transformation
$x\in\mathbb{R}^{N}\mapsto(r,\sigma)\in(0,\infty)\times\mathbb{S}^{N-1}$
(where $\mathbb{S}^{N-1}$ is the $\left(  N-1\right)  $-dimensional sphere
with respect to the Hausdorff measure in $\mathbb{R}^{N}$) and, for any function
$u\in C_{c}^{\infty}(\mathbb{R}^{N})$, we may expand $u$ in terms of spherical harmonics
\begin{equation}
  u(x)=u(r\sigma)=\sum_{k=0}^{\infty}\sum_{\ell=1}^{\dim\mathcal{H}_{k}}u_{k\ell}(r)\phi_{k\ell}(\sigma),\label{harmonics}%
\end{equation}
where $\{\phi_{k\ell}\}$, $k\geq0$, $\ell=1,\ldots,\mathcal{H}_{k}$, is an orthonormal basis in $L^{2}(\mathbb{S}%
^{N-1})$ constituted by spherical harmonic functions $\phi_{k\ell}$ of degree $k$, and $\mathcal{H}_{k}$ is the subspace of spherical harmonics of degree $k$.
Such $\phi_{k\ell}$ are smooth eigenfunctions of the Laplace-Beltrami operator
$-\Delta_{\mathbb{S}^{N-1}}$ with the corresponding eigenvalues $c_{k}%
=k(k+N-2)$, $k\geq0$, i.e.
\begin{equation}
  \begin{aligned}
-\Delta_{\mathbb{S}^{N-1}}\phi_{k\ell}=c_{k}\phi_{k\ell},\quad\left\langle \phi
_{k\ell},\phi_{m\ell'}\right\rangle _{L^{2}(\mathbb{S}^{N-1})}=\delta_{km}\delta_{\ell\ell'},&\quad\forall
k,m\in\mathbb{N},\\
&\ell=1,\ldots,\mathcal{H}_{k},\ell'=1,\ldots,\mathcal{H}_{m}
  \end{aligned}\label{eigen}%
\end{equation}
where $\left\langle \cdot,\cdot\right\rangle _{L^{2}(\mathbb{S}^{N-1})}$
denotes the scalar product in $L^{2}(\mathbb{S}^{N-1})$ and $\delta_{kl}$ is
the Kronecker symbol. 
In fact, without altering the following proofs and for sake of simplicity, we choose to use a single index and agree to write
\begin{equation}
  u(x)=u(r\sigma)=\sum_{k=0}^{\infty}u_{k}(r)\phi_{k}(\sigma),
\end{equation}
where $-\Delta_{\mathbb{S}^{N-1}}\phi_{k}=c_{k}\phi_{k}$.
See for example \cite{BEL,TZ,VZ}.

The Fourier coefficients $\{u_{k}\}_{k}$ belong to
$C_{c}^{\infty}([0,\infty))$ and satisfy $u_{k}(r)=O(r^{k})$, $u_{k}^{\prime
}(r)=O(r^{k-1})$ as $r\rightarrow0$. Also, the following formulas hold:
\begin{equation}
\Delta u=\mathcal{R}_{2}u+\frac{1}{r^{2}}\Delta_{\mathbb{S}^{N-1}}%
u,\quad|\nabla u|^{2}=|\partial_{r}u|^{2}+\frac{|\nabla_{\mathbb{S}^{N-1}%
}u|^{2}}{r^{2}}\label{Laplacian}%
\end{equation}
where, $\mathcal{R}_{2}u:=\partial_{rr}^{2}u+\frac{N-1}{r}\partial_{r}u$ is in
fact the radial Laplacian introduced before, $\partial_{r}$ and $\partial
_{rr}^{2}$ are both partial derivatives of first and second order with respect
to the radial component $r$, whereas $\Delta_{\mathbb{S}^{N-1}}$ and
$\nabla_{\mathbb{S}^{N-1}}$ represent the Laplace-Beltrami operator and
respectively the gradient operator with respect to the metric tensor on
$\mathbb{S}^{N-1}$ (for more details see, e.g. \cite{C, NLN}). We also note
that the behavior of $u_{k}(r)$ as $r\rightarrow0$ allows us to consider the
change of variables $u_{k}(r):=r^{k}v_{k}(r)$, where $v_{k}\in C_{c}^{\infty
}[0,\infty)$. 

Applying the co-area formula, properties above and \eqref{Laplacian}, we can
express both parts in \eqref{main1} in terms of the coefficients
$\{u_{k}\}_{k}$ and $\{v_{k}\}_{k}$. More precisely, we obtain the following identities:

\begin{lemma}
\label{l3.1}For $u\in C_{c}^{\infty}(\mathbb{R}^{N}):$%
\begin{align}
\int_{\mathbb{R}^{N}}|\nabla u|^{2}\mathrm{dx}  &  =\sum_{k=0}^{\infty}\left(
\int_{0}^{\infty}r^{N-1}|u_{k}^{\prime}\left(  r\right)  |^{2}dr+c_{k}\int
_{0}^{\infty}r^{N-3}\left\vert u_{k}\left(  r\right)  \right\vert
^{2}dr\right) \label{form1}\\
&  =\sum_{k=0}^{\infty}\int_{0}^{\infty}r^{N+2k-1}|v_{k}^{\prime}\left(
r\right)  |^{2}dr, \label{form1.1}%
\end{align}%
\begin{align}
\int_{\mathbb{R}^{N}}|x|^{2}|\nabla u|^{2}\mathrm{dx}  &  =\sum_{k=0}^{\infty
}\left(  \int_{0}^{\infty}r^{N+1}|u_{k}^{\prime}\left(  r\right)
|^{2}dr+c_{k}\int_{0}^{\infty}r^{N-1}\left\vert u_{k}\left(  r\right)
\right\vert ^{2}dr\right) \label{form2}\\
&  =\sum_{k=0}^{\infty}\left(  \int_{0}^{\infty}r^{N+2k+1}|v_{k}^{\prime
}\left(  r\right)  |^{2}dr-2k\int_{0}^{\infty}r^{N+2k-1}\left\vert
v_{k}\left(  r\right)  \right\vert ^{2}dr\right)  , \label{form2.1}%
\end{align}%
\begin{align}
\int_{\mathbb{R}^{N}}\frac{|\nabla u|^{2}}{\left\vert x\right\vert
}\mathrm{dx}  &  =\sum_{k=0}^{\infty}\left(  \int_{0}^{\infty}\left\vert
u_{k}^{\prime}\right\vert ^{2}r^{N-2}dr+c_{k}\int_{0}^{\infty}\left\vert
u_{k}\right\vert ^{2}r^{N-4}dr\right) \\
&  =\sum_{k=0}^{\infty}\left(  \int_{0}^{\infty}r^{N+2k-2}|v_{k}^{\prime}%
|^{2}dr+k\int_{0}^{\infty}r^{N+2k-4}\left\vert v_{k}\right\vert ^{2}dr\right)
,
\end{align}
and%
\begin{align}
\int_{\mathbb{R}^{N}}|\Delta u|^{2}\mathrm{dx}&=\sum_{k=0}^{\infty}\left(
\int_{0}^{\infty}r^{N+2k-1}|v_{k}^{\prime\prime}\left(  r\right)
|^{2}dr+(N+2k-1)\int_{0}^{\infty}r^{N+2k-3}|v_{k}^{\prime}\left(  r\right)
|^{2}dr\right)  . \label{form3}\\
&=\sum_{k=0}^{\infty} \Bigg(\int
_{0}^{\infty}r^{N-1}|u_{k}^{\prime\prime}(r)|^{2} dr + (N-1+2c_{k}) \int
_{0}^{\infty} r^{N-3} |u_{k}^{\prime}(r)|^{2} dr\nonumber\\
&+\left( c_{k}^{2} + 2c_{k} (N-4)\right) \int_{0}^{\infty}r^{N-5} |u_{k}(r)|^{2} dr\Bigg) \label{imp}
\end{align}

\end{lemma}

\begin{proof}
In view of \eqref{eigen} and \eqref{Laplacian}, we have
\begin{align*}
&  \int_{\mathbb{R}^{N}}|\nabla u|^{2}\mathrm{dx}=-\operatorname{Re}%
\int_{\mathbb{R}^{N}}\overline{u}\Delta u\mathrm{dx}\\
&  =-\operatorname{Re}\int_{0}^{\infty}\int_{\mathbb{S}^{N-1}}\left[
%TCIMACRO{\dsum \limits_{k=0}^{\infty}}%
%BeginExpansion
{\displaystyle\sum\limits_{k=0}^{\infty}}
%EndExpansion
\overline{u_{k}\left(  r\right)  }\overline{\phi_{k}\left(  \sigma\right)
}\right]  \left[
%TCIMACRO{\dsum \limits_{k=0}^{\infty}}%
%BeginExpansion
{\displaystyle\sum\limits_{k=0}^{\infty}}
%EndExpansion
\left(  u_{k}^{\prime\prime}\left(  r\right)  +\frac{N-1}{r}u_{k}^{\prime
}\left(  r\right)  -c_{k}\frac{u_{k}\left(  r\right)  }{r^{2}}\right)
\phi_{k}\left(  \sigma\right)  \right]  r^{N-1}drd\sigma\\
&  =-%
%TCIMACRO{\dsum \limits_{k=0}^{\infty}}%
%BeginExpansion
{\displaystyle\sum\limits_{k=0}^{\infty}}
%EndExpansion
\operatorname{Re}\int_{0}^{\infty}\overline{u_{k}\left(  r\right)  }\left(
u_{k}^{\prime\prime}\left(  r\right)  +\frac{N-1}{r}u_{k}^{\prime}\left(
r\right)  -c_{k}\frac{u_{k}\left(  r\right)  }{r^{2}}\right)  r^{N-1}dr\\
&  =-%
%TCIMACRO{\dsum \limits_{k=0}^{\infty}}%
%BeginExpansion
{\displaystyle\sum\limits_{k=0}^{\infty}}
%EndExpansion
\operatorname{Re}\left[  \int_{0}^{\infty}\overline{u_{k}\left(  r\right)
}u_{k}^{\prime\prime}\left(  r\right)  r^{N-1}dr+\left(  N-1\right)  \int
_{0}^{\infty}\overline{u_{k}\left(  r\right)  }u_{k}^{\prime}\left(  r\right)
r^{N-2}dr-c_{k}\int_{0}^{\infty}\left\vert u_{k}\left(  r\right)  \right\vert
^{2}r^{N-3}dr\right]  .
\end{align*}
Note that
\begin{align*}
-\operatorname{Re}\int_{0}^{\infty}\overline{u_{k}\left(  r\right)  }%
u_{k}^{\prime\prime}\left(  r\right)  r^{N-1}dr  &  =\operatorname{Re}\int
_{0}^{\infty}u_{k}^{\prime}\left(  r\right)  \left(  \overline{u_{k}\left(
r\right)  }r^{N-1}\right)  ^{\prime}dr\\
&  =\int_{0}^{\infty}\left\vert u_{k}^{\prime}\right\vert ^{2}r^{N-1}%
dr+\left(  N-1\right)  \operatorname{Re}\int_{0}^{\infty}u_{k}^{\prime}\left(
r\right)  \overline{u_{k}\left(  r\right)  }r^{N-2}dr.
\end{align*}
Therefore
\begin{align*}
&  \int_{\mathbb{R}^{N}}|\nabla u|^{2}\mathrm{dx}\\
&  =%
%TCIMACRO{\dsum \limits_{k=0}^{\infty}}%
%BeginExpansion
{\displaystyle\sum\limits_{k=0}^{\infty}}
%EndExpansion
\left(  \int_{0}^{\infty}\left\vert u_{k}^{\prime}\right\vert ^{2}%
r^{N-1}dr+c_{k}\int_{0}^{\infty}\left\vert u_{k}\right\vert ^{2}%
r^{N-3}dr\right) \\
&  =\sum_{k=0}^{\infty}\left(  \int_{0}^{\infty}r^{N-1}|(r^{k}v_{k})^{\prime
}|^{2}dr+c_{k}\int_{0}^{\infty}r^{N-3}\left\vert r^{k}v_{k}\right\vert
^{2}dr\right) \\
&  =\sum_{k=0}^{\infty}\left(  \int_{0}^{\infty}r^{N-1}|kr^{k-1}v_{k}%
+r^{k}v_{k}^{\prime}|^{2}dr+c_{k}\int_{0}^{\infty}r^{N+2k-3}\left\vert
v_{k}\right\vert ^{2}dr\right) \\
&  =\sum_{k=0}^{\infty}\left(  \int_{0}^{\infty}\left(  r^{N+2k-1}%
|v_{k}^{\prime}|^{2}+k^{2}r^{N+2k-3}\left\vert v_{k}\right\vert ^{2}%
+kr^{N+2k-2}(\left\vert v_{k}\right\vert ^{2})^{\prime}\right)  dr+c_{k}%
\int_{0}^{\infty}r^{N+2k-3}\left\vert v_{k}\right\vert ^{2}dr\right) \\
&  =\sum_{k=0}^{\infty}\int_{0}^{\infty}r^{N+2k-1}|v_{k}^{\prime}|^{2}dr.
\end{align*}
Here the last identity follows by integrating by parts the terms $k\int
_{0}^{\infty}r^{N+2k-2}(\left\vert v_{k}\right\vert ^{2})^{\prime}dr:$%
\[
k\int_{0}^{\infty}r^{N+2k-2}(\left\vert v_{k}\right\vert ^{2})^{\prime
}dr=-k\left(  N+2k-2\right)  \int_{0}^{\infty}\left\vert v_{k}\right\vert
^{2}r^{N+2k-3}dr.
\]
Similarly, we write successively
\begin{align*}
\int_{\mathbb{R}^{N}}\left\vert x\right\vert ^{2}|\nabla u|^{2}\mathrm{dx}  &
=-\operatorname{Re}\int_{\mathbb{R}^{N}}\overline{u}\operatorname{div}\left(
\left\vert x\right\vert ^{2}\nabla u\right)  \mathrm{dx}\\
&  =-\operatorname{Re}\int_{\mathbb{R}^{N}}\left\vert x\right\vert
^{2}\overline{u}\Delta u\mathrm{dx}-2\operatorname{Re}\int_{\mathbb{R}^{N}%
}\overline{u}x\cdot\nabla u\mathrm{dx}\\
&  =-\operatorname{Re}\int_{\mathbb{R}^{N}}\left\vert x\right\vert
^{2}\overline{u}\Delta u\mathrm{dx}-\int_{\mathbb{R}^{N}}x\cdot\nabla
(\left\vert u\right\vert ^{2})\mathrm{dx}\\
&  =-\operatorname{Re}\int_{\mathbb{R}^{N}}\left\vert x\right\vert
^{2}\overline{u}\Delta u\mathrm{dx}+N\int_{\mathbb{R}^{N}}\left\vert
u\right\vert ^{2}\mathrm{dx}.
\end{align*}
Note%
\begin{align*}
&  \operatorname{Re}\int_{\mathbb{R}^{N}}\left\vert x\right\vert ^{2}%
\overline{u}\Delta u\mathrm{dx}\\
&  =\operatorname{Re}\int_{0}^{\infty}\int_{\mathbb{S}^{N-1}}\left[
%TCIMACRO{\dsum \limits_{k=0}^{\infty}}%
%BeginExpansion
{\displaystyle\sum\limits_{k=0}^{\infty}}
%EndExpansion
\overline{u_{k}\left(  r\right)  }\overline{\phi_{k}\left(  \sigma\right)
}\right]  \left[
%TCIMACRO{\dsum \limits_{k=0}^{\infty}}%
%BeginExpansion
{\displaystyle\sum\limits_{k=0}^{\infty}}
%EndExpansion
\left(  u_{k}^{\prime\prime}\left(  r\right)  +\frac{N-1}{r}u_{k}^{\prime
}\left(  r\right)  -c_{k}\frac{u_{k}\left(  r\right)  }{r^{2}}\right)
\phi_{k}\left(  \sigma\right)  \right]  r^{N+1}drd\sigma\\
&  =%
%TCIMACRO{\dsum \limits_{k=0}^{\infty}}%
%BeginExpansion
{\displaystyle\sum\limits_{k=0}^{\infty}}
%EndExpansion
\operatorname{Re}\int_{0}^{\infty}\overline{u_{k}\left(  r\right)  }\left(
u_{k}^{\prime\prime}\left(  r\right)  +\frac{N-1}{r}u_{k}^{\prime}\left(
r\right)  -c_{k}\frac{u_{k}\left(  r\right)  }{r^{2}}\right)  r^{N+1}dr\\
&  =%
%TCIMACRO{\dsum \limits_{k=0}^{\infty}}%
%BeginExpansion
{\displaystyle\sum\limits_{k=0}^{\infty}}
%EndExpansion
\left[  \operatorname{Re}\int_{0}^{\infty}\overline{u_{k}\left(  r\right)
}u_{k}^{\prime\prime}\left(  r\right)  r^{N+1}dr+\left(  N-1\right)
\operatorname{Re}\int_{0}^{\infty}\overline{u_{k}\left(  r\right)  }%
u_{k}^{\prime}\left(  r\right)  r^{N}dr-c_{k}\int_{0}^{\infty}\left\vert
u_{k}\right\vert ^{2}r^{N-1}dr\right] \\
&  =%
%TCIMACRO{\dsum \limits_{k=0}^{\infty}}%
%BeginExpansion
{\displaystyle\sum\limits_{k=0}^{\infty}}
%EndExpansion
\left[  -\int_{0}^{\infty}\left\vert u_{k}^{\prime}\right\vert ^{2}%
r^{N+1}dr-2\operatorname{Re}\int_{0}^{\infty}u_{k}^{\prime}\left(  r\right)
\overline{u_{k}\left(  r\right)  }r^{N}dr-c_{k}\int_{0}^{\infty}\left\vert
u_{k}\right\vert ^{2}r^{N-1}dr\right] \\
&  ={\sum\limits_{k=0}^{\infty}}\left[  -\int_{0}^{\infty}\left\vert
u_{k}^{\prime}\right\vert ^{2}r^{N+1}dr+N\int_{0}^{\infty}|u_{k}%
(r)|^{2}r^{N-1}dr-c_{k}\int_{0}^{\infty}\left\vert u_{k}\right\vert
^{2}r^{N-1}dr\right]
\end{align*}
and%
\begin{align*}
\int_{\mathbb{R}^{N}}\left\vert u\right\vert ^{2}\mathrm{dx}  &  =\int
_{0}^{\infty}\int_{\mathbb{S}^{N-1}}\left\vert
%TCIMACRO{\dsum \limits_{k=0}^{\infty}}%
%BeginExpansion
{\displaystyle\sum\limits_{k=0}^{\infty}}
%EndExpansion
u_{k}\left(  r\right)  \phi_{k}\left(  \sigma\right)  \right\vert ^{2}%
r^{N-1}d\sigma dr\\
&  =\sum_{k=0}^{\infty}\int_{0}^{\infty}|u_{k}|^{2}r^{N-1}dr.
\end{align*}
Therefore
\begin{align*}
&  \int_{\mathbb{R}^{N}}\left\vert x\right\vert ^{2}|\nabla u|^{2}%
\mathrm{dx}\\
&  =%
%TCIMACRO{\dsum \limits_{k=0}^{\infty}}%
%BeginExpansion
{\displaystyle\sum\limits_{k=0}^{\infty}}
%EndExpansion
\left(  \int_{0}^{\infty}\left\vert u_{k}^{\prime}\right\vert ^{2}%
r^{N+1}dr+c_{k}\int_{0}^{\infty}\left\vert u_{k}\right\vert ^{2}%
r^{N-1}dr\right) \\
&  =\sum_{k=0}^{\infty}\left(  \int_{0}^{\infty}r^{N+1}|(r^{k}v_{k})^{\prime
}|^{2}dr+c_{k}\int_{0}^{\infty}r^{N-1}\left\vert r^{k}v_{k}\right\vert
^{2}dr\right) \\
&  =\sum_{k=0}^{\infty}\left(  \int_{0}^{\infty}r^{N+1}|kr^{k-1}v_{k}%
+r^{k}v_{k}^{\prime}|^{2}dr+c_{k}\int_{0}^{\infty}r^{N+2k-1}\left\vert
v_{k}\right\vert ^{2}dr\right) \\
&  =\sum_{k=0}^{\infty}\left(  \int_{0}^{\infty}\left(  r^{N+2k+1}%
|v_{k}^{\prime}|^{2}+k^{2}r^{N+2k-1}\left\vert v_{k}\right\vert ^{2}%
+kr^{N+2k}(\left\vert v_{k}\right\vert ^{2})^{\prime}\right)  dr+c_{k}\int
_{0}^{\infty}r^{N+2k-1}\left\vert v_{k}\right\vert ^{2}dr\right) \\
&  =\sum_{k=0}^{\infty}\left(  \int_{0}^{\infty}r^{N+2k+1}|v_{k}^{\prime}%
|^{2}dr-2k\int_{0}^{\infty}r^{N+2k-1}\left\vert v_{k}\right\vert
^{2}dr\right)  .
\end{align*}
Here the last identity follows by integrating by parts the terms $k\int
_{0}^{\infty}r^{N+2k}(\left\vert v_{k}\right\vert ^{2})^{\prime}dr:$%
\[
k\int_{0}^{\infty}r^{N+2k}(\left\vert v_{k}\right\vert ^{2})^{\prime
}dr=-k\left(  N+2k\right)  \int_{0}^{\infty}\left\vert v_{k}\right\vert
^{2}r^{N+2k-3}dr.
\]
We also have that%
\begin{align*}
\int_{\mathbb{R}^{N}}\frac{|\nabla u|^{2}}{\left\vert x\right\vert
}\mathrm{dx}  &  =-\operatorname{Re}\int_{\mathbb{R}^{N}}\overline
{u}\operatorname{div}\left(  \frac{\nabla u}{\left\vert x\right\vert }\right)
\mathrm{dx}\\
&  =-\operatorname{Re}\int_{\mathbb{R}^{N}}\frac{1}{\left\vert x\right\vert
}\overline{u}\Delta u\mathrm{dx}+\operatorname{Re}\int_{\mathbb{R}^{N}%
}\overline{u}\frac{x}{\left\vert x\right\vert ^{3}}\cdot\nabla u\mathrm{dx}\\
&  =-\operatorname{Re}\int_{\mathbb{R}^{N}}\frac{1}{\left\vert x\right\vert
}\overline{u}\Delta u\mathrm{dx}+\frac{1}{2}\int_{\mathbb{R}^{N}}\frac
{x}{\left\vert x\right\vert ^{3}}\cdot\nabla(\left\vert u\right\vert
^{2})\mathrm{dx}\\
&  =-\operatorname{Re}\int_{\mathbb{R}^{N}}\frac{1}{\left\vert x\right\vert
}\overline{u}\Delta u\mathrm{dx}-\frac{N-3}{2}\int_{\mathbb{R}^{N}}%
\frac{\left\vert u\right\vert ^{2}}{\left\vert x\right\vert ^{3}}\mathrm{dx}.
\end{align*}
Note%
\begin{align*}
&  \operatorname{Re}\int_{\mathbb{R}^{N}}\frac{1}{\left\vert x\right\vert
}\overline{u}\Delta u\mathrm{dx}\\
&  =\operatorname{Re}\int_{0}^{\infty}\int_{\mathbb{S}^{N-1}}\left[
%TCIMACRO{\dsum \limits_{k=0}^{\infty}}%
%BeginExpansion
{\displaystyle\sum\limits_{k=0}^{\infty}}
%EndExpansion
\overline{u_{k}\left(  r\right)  }\overline{\phi_{k}\left(  \sigma\right)
}\right]  \left[
%TCIMACRO{\dsum \limits_{k=0}^{\infty}}%
%BeginExpansion
{\displaystyle\sum\limits_{k=0}^{\infty}}
%EndExpansion
\left(  u_{k}^{\prime\prime}\left(  r\right)  +\frac{N-1}{r}u_{k}^{\prime
}\left(  r\right)  -c_{k}\frac{u_{k}\left(  r\right)  }{r^{2}}\right)
\phi_{k}\left(  \sigma\right)  \right]  r^{N-2}drd\sigma\\
&  =%
%TCIMACRO{\dsum \limits_{k=0}^{\infty}}%
%BeginExpansion
{\displaystyle\sum\limits_{k=0}^{\infty}}
%EndExpansion
\operatorname{Re}\int_{0}^{\infty}\overline{u_{k}\left(  r\right)  }\left(
u_{k}^{\prime\prime}\left(  r\right)  +\frac{N-1}{r}u_{k}^{\prime}\left(
r\right)  -c_{k}\frac{u_{k}\left(  r\right)  }{r^{2}}\right)  r^{N-2}dr\\
&  =%
%TCIMACRO{\dsum \limits_{k=0}^{\infty}}%
%BeginExpansion
{\displaystyle\sum\limits_{k=0}^{\infty}}
%EndExpansion
\left[  \operatorname{Re}\int_{0}^{\infty}\overline{u_{k}\left(  r\right)
}u_{k}^{\prime\prime}\left(  r\right)  r^{N-2}dr+\left(  N-1\right)
\operatorname{Re}\int_{0}^{\infty}\overline{u_{k}\left(  r\right)  }%
u_{k}^{\prime}\left(  r\right)  r^{N-3}dr-c_{k}\int_{0}^{\infty}\left\vert
u_{k}\right\vert ^{2}r^{N-4}dr\right] \\
&  =%
%TCIMACRO{\dsum \limits_{k=0}^{\infty}}%
%BeginExpansion
{\displaystyle\sum\limits_{k=0}^{\infty}}
%EndExpansion
\left[  -\int_{0}^{\infty}\left\vert u_{k}^{\prime}\right\vert ^{2}%
r^{N-2}dr+\operatorname{Re}\int_{0}^{\infty}u_{k}^{\prime}\left(  r\right)
\overline{u_{k}\left(  r\right)  }r^{N-3}dr-c_{k}\int_{0}^{\infty}\left\vert
u_{k}\right\vert ^{2}r^{N-4}dr\right] \\
&  ={\sum\limits_{k=0}^{\infty}}\left[  -\int_{0}^{\infty}\left\vert
u_{k}^{\prime}\right\vert ^{2}r^{N-2}dr-\frac{N-3}{2}\int_{0}^{\infty}%
|u_{k}(r)|^{2}r^{N-4}dr-c_{k}\int_{0}^{\infty}\left\vert u_{k}\right\vert
^{2}r^{N-4}dr\right]
\end{align*}
and%
\begin{align*}
\int_{\mathbb{R}^{N}}\frac{\left\vert u\right\vert ^{2}}{\left\vert
x\right\vert ^{3}}\mathrm{dx}  &  =\int_{0}^{\infty}\int_{\mathbb{S}^{N-1}%
}\left\vert
%TCIMACRO{\dsum \limits_{k=0}^{\infty}}%
%BeginExpansion
{\displaystyle\sum\limits_{k=0}^{\infty}}
%EndExpansion
u_{k}\left(  r\right)  \phi_{k}\left(  \sigma\right)  \right\vert ^{2}%
r^{N-4}d\sigma dr\\
&  =\sum_{k=0}^{\infty}\int_{0}^{\infty}|u_{k}|^{2}r^{N-4}dr.
\end{align*}
Therefore%
\begin{align*}
&  \int_{\mathbb{R}^{N}}\frac{|\nabla u|^{2}}{\left\vert x\right\vert
}\mathrm{dx}\\
&  =%
%TCIMACRO{\dsum \limits_{k=0}^{\infty}}%
%BeginExpansion
{\displaystyle\sum\limits_{k=0}^{\infty}}
%EndExpansion
\left(  \int_{0}^{\infty}\left\vert u_{k}^{\prime}\right\vert ^{2}%
r^{N-2}dr+c_{k}\int_{0}^{\infty}\left\vert u_{k}\right\vert ^{2}%
r^{N-4}dr\right) \\
&  =\sum_{k=0}^{\infty}\left(  \int_{0}^{\infty}r^{N-2}|(r^{k}v_{k})^{\prime
}|^{2}dr+c_{k}\int_{0}^{\infty}r^{N-4}\left\vert r^{k}v_{k}\right\vert
^{2}dr\right) \\
&  =\sum_{k=0}^{\infty}\left(  \int_{0}^{\infty}r^{N-2}|kr^{k-1}v_{k}%
+r^{k}v_{k}^{\prime}|^{2}dr+c_{k}\int_{0}^{\infty}r^{N+2k-4}\left\vert
v_{k}\right\vert ^{2}dr\right) \\
&  =\sum_{k=0}^{\infty}\left(  \int_{0}^{\infty}\left(  r^{N+2k-2}%
|v_{k}^{\prime}|^{2}+k^{2}r^{N+2k-4}\left\vert v_{k}\right\vert ^{2}%
+kr^{N+2k-3}(\left\vert v_{k}\right\vert ^{2})^{\prime}\right)  dr+c_{k}%
\int_{0}^{\infty}r^{N+2k-4}\left\vert v_{k}\right\vert ^{2}dr\right) \\
&  =\sum_{k=0}^{\infty}\left(  \int_{0}^{\infty}r^{N+2k-2}|v_{k}^{\prime}%
|^{2}dr+k\int_{0}^{\infty}r^{N+2k-4}\left\vert v_{k}\right\vert ^{2}dr\right)
.
\end{align*}

Finally, we get%
\begin{align*}
&  \int_{\mathbb{R}^{N}}\left\vert \Delta u\right\vert ^{2}\mathrm{dx}\\
&  =\int_{0}^{\infty}\int_{\mathbb{S}^{N-1}}\left\vert
%TCIMACRO{\dsum \limits_{k=0}^{\infty}}%
%BeginExpansion
{\displaystyle\sum\limits_{k=0}^{\infty}}
%EndExpansion
\Delta\left(  u_{k}\left(  r\right)  \phi_{k}\left(  \sigma\right)  \right)
\right\vert ^{2}r^{N-1}d\sigma dr\\
&  =\int_{0}^{\infty}\int_{\mathbb{S}^{N-1}}\left\vert
%TCIMACRO{\dsum \limits_{k=0}^{\infty}}%
%BeginExpansion
{\displaystyle\sum\limits_{k=0}^{\infty}}
%EndExpansion
u_{k}^{\prime\prime}\left(  r\right)  \phi_{k}\left(  \sigma\right)
+\frac{N-1}{r}u_{k}^{\prime}\left(  r\right)  \phi_{k}\left(  \sigma\right)
-c_{k}\frac{u_{k}\left(  r\right)  }{r^{2}}\phi_{k}\left(  \sigma\right)
\right\vert ^{2}r^{N-1}d\sigma dr\\
&  =%
%TCIMACRO{\dsum \limits_{k=0}^{\infty}}%
%BeginExpansion
{\displaystyle\sum\limits_{k=0}^{\infty}}
%EndExpansion
\int_{0}^{\infty}\left\vert u_{k}^{\prime\prime}\left(  r\right)  +\frac
{N-1}{r}u_{k}^{\prime}\left(  r\right)  -c_{k}\frac{u_{k}\left(  r\right)
}{r^{2}}\right\vert ^{2}r^{N-1}dr\\
&  =\sum_{k=0}^{\infty}\int_{0}^{\infty}r^{N-1}\left\vert (r^{k}%
v_{k}(r))^{\prime\prime}+\frac{N-1}{r}(r^{k}v_{k}(r))^{\prime}-c_{k}%
r^{k-2}v_{k}(r)\right\vert ^{2}dr\\
&  =\sum_{k=0}^{\infty}\int_{0}^{\infty}r^{N-1}\left\vert r^{k}v_{k}%
^{\prime\prime}(r)+(N+2k-1)r^{k-1}v_{k}^{\prime}(r)\right\vert ^{2}dr\\
&  =\sum_{k=0}^{\infty}\Big(\int_{0}^{\infty}r^{N+2k-1}|v_{k}^{\prime\prime
}(r)|^{2}dr+(N+2k-1)^{2}\int_{0}^{\infty}r^{N+2k-3}|v_{k}^{\prime}(r)|^{2}dr\\
&  \phantom{XXXXXXXXXXXXxXX}+(N+2k-1)\int_{0}^{\infty}r^{N+2k-2}\left(
\left\vert v_{k}^{\prime}(r)\right\vert ^{2}\right)  ^{\prime}dr\Big)\\
&  =\sum_{k=0}^{\infty}\left(  \int_{0}^{\infty}r^{N+2k-1}|v_{k}^{\prime
\prime}(r)|^{2}dr+(N+2k-1)\int_{0}^{\infty}r^{N+2k-3}|v_{k}^{\prime}%
(r)|^{2}dr\right)  .
\end{align*}
Here the last identity follows from the following identity:%
\[
\int_{0}^{\infty}r^{N+2k-2}\left(  |v_{k}^{\prime}(r)|^{2}\right)  ^{\prime
}dr=-\left(  N+2k-2\right)  \int_{0}^{\infty}r^{N+2k-3}|v_{k}^{\prime}%
(r)|^{2}dr.
\]
We note that in the case $N=2$ and $k=0$, the above identity requires the
condition $v_{0}^{\prime}\left(  0\right)  =0$. However, we can get this
condition from the following argument: since%
\begin{multline*}
\int_{\mathbb{R}^{2}}|\Delta u|^{2}\mathrm{dx}=\sum_{k=0}^{\infty}%
\Big(\int_{0}^{\infty}r^{2+2k-1}|v_{k}^{\prime\prime}(r)|^{2}dr+(2+2k-1)^{2}%
\int_{0}^{\infty}r^{2+2k-3}|v_{k}^{\prime}(r)|^{2}dr\\
+(2+2k-1)\int_{0}^{\infty}r^{2+2k-2}\left(  \left\vert v_{k}^{\prime
}(r)\right\vert ^{2}\right)  ^{\prime}dr\Big),
\end{multline*}
we deduce that $\int_{0}^{\infty}r^{-1}|v_{0}^{\prime}(r)|^{2}dr$ is
well-defined since all the other terms are well-defined. Therefore we get
$v_{0}^{\prime}(0)=0$.

Identity \eqref{imp} can be obtained similarly as \eqref{form3}. We let the details to the reader since \eqref{imp} is an auxiliary result for our purpose, being applied only to argue Remark \ref{r2.2}. 
\end{proof}

For the radial operators $\mathcal{R}_{1}$ and $\mathcal{R}_{2}$, we also have the following analogous identities. 

\begin{lemma}
\label{l3.2}For $u\in C_{c}^{\infty}(\mathbb{R}^{N}):$%
\begin{align*}
\int_{\mathbb{R}^{N}}|\mathcal{R}_{1}u|^{2}\mathrm{dx}  &  =%
%TCIMACRO{\dsum \limits_{k=0}^{\infty}}%
%BeginExpansion
{\displaystyle\sum\limits_{k=0}^{\infty}}
%EndExpansion
\int_{0}^{\infty}\left\vert u_{k}^{\prime}\left(  r\right)  \right\vert
^{2}r^{N-1}dr\\
&  =%
%TCIMACRO{\dsum \limits_{k=0}^{\infty}}%
%BeginExpansion
{\displaystyle\sum\limits_{k=0}^{\infty}}
%EndExpansion
\left(  \int_{0}^{\infty}r^{N+2k-1}\left\vert v_{k}^{\prime}\left(  r\right)
\right\vert ^{2}dr-c_{k}\int_{0}^{\infty}r^{N+2k-3}\left\vert v_{k}\left(
r\right)  \right\vert ^{2}dr\right)  ,
\end{align*}%
\begin{align*}
\int_{\mathbb{R}^{N}}\left\vert x\right\vert ^{2}|\mathcal{R}_{1}%
u|^{2}\mathrm{dx}  &  =%
%TCIMACRO{\dsum \limits_{k=0}^{\infty}}%
%BeginExpansion
{\displaystyle\sum\limits_{k=0}^{\infty}}
%EndExpansion
\int_{0}^{\infty}\left\vert u_{k}^{\prime}\left(  r\right)  \right\vert
^{2}r^{N+1}dr\\
&  =%
%TCIMACRO{\dsum \limits_{k=0}^{\infty}}%
%BeginExpansion
{\displaystyle\sum\limits_{k=0}^{\infty}}
%EndExpansion
\left(  \int_{0}^{\infty}r^{N+2k+1}\left\vert v_{k}^{\prime}\left(  r\right)
\right\vert ^{2}dr-\left(  c_{k}+2k\right)  \int_{0}^{\infty}r^{N+2k-1}%
\left\vert v_{k}\left(  r\right)  \right\vert ^{2}dr\right)  ,
\end{align*}%
\begin{align*}
\int_{\mathbb{R}^{N}}\frac{|\mathcal{R}_{1}u|^{2}}{\left\vert x\right\vert
}\mathrm{dx}  &  =%
%TCIMACRO{\dsum \limits_{k=0}^{\infty}}%
%BeginExpansion
{\displaystyle\sum\limits_{k=0}^{\infty}}
%EndExpansion
\int_{0}^{\infty}\left\vert u_{k}^{\prime}\left(  r\right)  \right\vert
^{2}r^{N-2}dr\\
&  =%
%TCIMACRO{\dsum \limits_{k=0}^{\infty}}%
%BeginExpansion
{\displaystyle\sum\limits_{k=0}^{\infty}}
%EndExpansion
\left(  \int_{0}^{\infty}r^{N+2k-2}\left\vert v_{k}^{\prime}\left(  r\right)
\right\vert ^{2}dr-\left(  c_{k}-k\right)  \int_{0}^{\infty}r^{N+2k-4}%
\left\vert v_{k}\left(  r\right)  \right\vert ^{2}dr\right)  ,
\end{align*}
and%
\[
\int_{\mathbb{R}^{N}}|\mathcal{R}_{2}u|^{2}\mathrm{dx}=%
%TCIMACRO{\dsum \limits_{k=0}^{\infty}}%
%BeginExpansion
{\displaystyle\sum\limits_{k=0}^{\infty}}
%EndExpansion
\left(  \int_{0}^{\infty}\left\vert u_{k}^{\prime\prime}\left(  r\right)
\right\vert ^{2}r^{N-1}dr+\left(  N-1\right)  \int_{0}^{\infty}\left\vert
u_{k}^{\prime}\left(  r\right)  \right\vert ^{2}r^{N-3}dr\right)  .
\]

\end{lemma}

\begin{proof}
We have%
\begin{align*}
&  \int_{\mathbb{R}^{N}}\left\vert \mathcal{R}_{1}u\right\vert ^{2}%
\mathrm{dx}\\
&  =\int_{0}^{\infty}\int_{\mathbb{S}^{N-1}}\left\vert
%TCIMACRO{\dsum \limits_{k=0}^{\infty}}%
%BeginExpansion
{\displaystyle\sum\limits_{k=0}^{\infty}}
%EndExpansion
u_{k}^{\prime}\left(  r\right)  \phi_{k}\left(  \sigma\right)  \right\vert
^{2}r^{N-1}drd\sigma\\
&  =%
%TCIMACRO{\dsum \limits_{k=0}^{\infty}}%
%BeginExpansion
{\displaystyle\sum\limits_{k=0}^{\infty}}
%EndExpansion
\int_{0}^{\infty}\left\vert u_{k}^{\prime}\left(  r\right)  \right\vert
^{2}r^{N-1}dr\\
&  =%
%TCIMACRO{\dsum \limits_{k=0}^{\infty}}%
%BeginExpansion
{\displaystyle\sum\limits_{k=0}^{\infty}}
%EndExpansion
\int_{0}^{\infty}\left\vert \left(  r^{k}v_{k}\left(  r\right)  \right)
^{\prime}\right\vert ^{2}r^{N-1}dr\\
&  =%
%TCIMACRO{\dsum \limits_{k=0}^{\infty}}%
%BeginExpansion
{\displaystyle\sum\limits_{k=0}^{\infty}}
%EndExpansion
\int_{0}^{\infty}\left\vert r^{k}v_{k}^{\prime}\left(  r\right)
+kr^{k-1}v_{k}\left(  r\right)  \right\vert ^{2}r^{N-1}dr\\
&  =%
%TCIMACRO{\dsum \limits_{k=0}^{\infty}}%
%BeginExpansion
{\displaystyle\sum\limits_{k=0}^{\infty}}
%EndExpansion
\left(  \int_{0}^{\infty}r^{N+2k-1}\left\vert v_{k}^{\prime}\left(  r\right)
\right\vert ^{2}dr+k\int_{0}^{\infty}r^{N+2k-2}\left(  \left\vert v_{k}\left(
r\right)  \right\vert ^{2}\right)  ^{\prime}dr+k^{2}\int_{0}^{\infty
}r^{N+2k-3}\left\vert v_{k}\left(  r\right)  \right\vert ^{2}dr\right)  .
\end{align*}
Using the identity%
\[
k\int_{0}^{\infty}r^{N+2k-2}\left(  \left\vert v_{k}\left(  r\right)
\right\vert ^{2}\right)  ^{\prime}dr=-k\left(  N+2k-2\right)  \int_{0}%
^{\infty}r^{N+2k-3}\left\vert v_{k}\left(  r\right)  \right\vert ^{2}dr,
\]
we get%
\[
\int_{\mathbb{R}^{N}}\left\vert \mathcal{R}_{1}u\right\vert ^{2}\mathrm{dx}=%
%TCIMACRO{\dsum \limits_{k=0}^{\infty}}%
%BeginExpansion
{\displaystyle\sum\limits_{k=0}^{\infty}}
%EndExpansion
\left(  \int_{0}^{\infty}r^{N+2k-1}\left\vert v_{k}^{\prime}\left(  r\right)
\right\vert ^{2}dr-c_{k}\int_{0}^{\infty}r^{N+2k-3}\left\vert v_{k}\left(
r\right)  \right\vert ^{2}dr\right)  .
\]
Similarly,%
\begin{align*}
&  \int_{\mathbb{R}^{N}}\left\vert x\right\vert ^{2}\left\vert \mathcal{R}%
_{1}u\right\vert ^{2}\mathrm{dx}\\
&  =\int_{0}^{\infty}\int_{\mathbb{S}^{N-1}}\left\vert
%TCIMACRO{\dsum \limits_{k=0}^{\infty}}%
%BeginExpansion
{\displaystyle\sum\limits_{k=0}^{\infty}}
%EndExpansion
u_{k}^{\prime}\left(  r\right)  \phi_{k}\left(  \sigma\right)  \right\vert
^{2}r^{N+1}drd\sigma\\
&  =%
%TCIMACRO{\dsum \limits_{k=0}^{\infty}}%
%BeginExpansion
{\displaystyle\sum\limits_{k=0}^{\infty}}
%EndExpansion
\int_{0}^{\infty}\left\vert u_{k}^{\prime}\left(  r\right)  \right\vert
^{2}r^{N+1}dr\\
&  =%
%TCIMACRO{\dsum \limits_{k=0}^{\infty}}%
%BeginExpansion
{\displaystyle\sum\limits_{k=0}^{\infty}}
%EndExpansion
\int_{0}^{\infty}\left\vert \left(  r^{k}v_{k}\left(  r\right)  \right)
^{\prime}\right\vert ^{2}r^{N+1}dr\\
&  =%
%TCIMACRO{\dsum \limits_{k=0}^{\infty}}%
%BeginExpansion
{\displaystyle\sum\limits_{k=0}^{\infty}}
%EndExpansion
\int_{0}^{\infty}\left\vert r^{k}v_{k}^{\prime}\left(  r\right)
+kr^{k-1}v_{k}\left(  r\right)  \right\vert ^{2}r^{N+1}dr\\
&  =%
%TCIMACRO{\dsum \limits_{k=0}^{\infty}}%
%BeginExpansion
{\displaystyle\sum\limits_{k=0}^{\infty}}
%EndExpansion
\left(  \int_{0}^{\infty}r^{N+2k+1}\left\vert v_{k}^{\prime}\left(  r\right)
\right\vert ^{2}dr+k\int_{0}^{\infty}r^{N+2k}\left(  \left\vert v_{k}\left(
r\right)  \right\vert ^{2}\right)  ^{\prime}dr+k^{2}\int_{0}^{\infty
}r^{N+2k-1}\left\vert v_{k}\left(  r\right)  \right\vert ^{2}dr\right)  .
\end{align*}
Using the identity%
\[
k\int_{0}^{\infty}r^{N+2k}\left(  \left\vert v_{k}\left(  r\right)
\right\vert ^{2}\right)  ^{\prime}dr=-k\left(  N+2k\right)  \int_{0}^{\infty
}r^{N+2k-1}\left\vert v_{k}\left(  r\right)  \right\vert ^{2}dr,
\]
we get%
\[
\int_{\mathbb{R}^{N}}\left\vert x\right\vert ^{2}\left\vert \mathcal{R}%
_{1}u\right\vert ^{2}\mathrm{dx}=%
%TCIMACRO{\dsum \limits_{k=0}^{\infty}}%
%BeginExpansion
{\displaystyle\sum\limits_{k=0}^{\infty}}
%EndExpansion
\left(  \int_{0}^{\infty}r^{N+2k+1}\left\vert v_{k}^{\prime}\left(  r\right)
\right\vert ^{2}dr-\left(  c_{k}+2k\right)  \int_{0}^{\infty}r^{N+2k-1}%
\left\vert v_{k}\left(  r\right)  \right\vert ^{2}dr\right)  .
\]
We also have%
\begin{align*}
&  \int_{\mathbb{R}^{N}}\frac{|\mathcal{R}_{1}u|^{2}}{\left\vert x\right\vert
}\mathrm{dx}\\
&  =\int_{0}^{\infty}\int_{\mathbb{S}^{N-1}}\left\vert
%TCIMACRO{\dsum \limits_{k=0}^{\infty}}%
%BeginExpansion
{\displaystyle\sum\limits_{k=0}^{\infty}}
%EndExpansion
u_{k}^{\prime}\left(  r\right)  \phi_{k}\left(  \sigma\right)  \right\vert
^{2}r^{N-2}drd\sigma\\
&  =%
%TCIMACRO{\dsum \limits_{k=0}^{\infty}}%
%BeginExpansion
{\displaystyle\sum\limits_{k=0}^{\infty}}
%EndExpansion
\int_{0}^{\infty}\left\vert u_{k}^{\prime}\left(  r\right)  \right\vert
^{2}r^{N-2}dr\\
&  =%
%TCIMACRO{\dsum \limits_{k=0}^{\infty}}%
%BeginExpansion
{\displaystyle\sum\limits_{k=0}^{\infty}}
%EndExpansion
\int_{0}^{\infty}\left\vert \left(  r^{k}v_{k}\left(  r\right)  \right)
^{\prime}\right\vert ^{2}r^{N-2}dr\\
&  =%
%TCIMACRO{\dsum \limits_{k=0}^{\infty}}%
%BeginExpansion
{\displaystyle\sum\limits_{k=0}^{\infty}}
%EndExpansion
\int_{0}^{\infty}\left\vert r^{k}v_{k}^{\prime}\left(  r\right)
+kr^{k-1}v_{k}\left(  r\right)  \right\vert ^{2}r^{N-2}dr\\
&  =%
%TCIMACRO{\dsum \limits_{k=0}^{\infty}}%
%BeginExpansion
{\displaystyle\sum\limits_{k=0}^{\infty}}
%EndExpansion
\left(  \int_{0}^{\infty}r^{N+2k-2}\left\vert v_{k}^{\prime}\left(  r\right)
\right\vert ^{2}dr+k\int_{0}^{\infty}r^{N+2k-3}\left(  \left\vert v_{k}\left(
r\right)  \right\vert ^{2}\right)  ^{\prime}dr+k^{2}\int_{0}^{\infty
}r^{N+2k-4}\left\vert v_{k}\left(  r\right)  \right\vert ^{2}dr\right)  .
\end{align*}
By integration by parts, we get
\[
k\int_{0}^{\infty}r^{N+2k-3}\left(  \left\vert v_{k}\left(  r\right)
\right\vert ^{2}\right)  ^{\prime}dr=-k\left(  N+2k-3\right)  \int_{0}%
^{\infty}r^{N+2k-4}\left\vert v_{k}\left(  r\right)  \right\vert ^{2}dr.
\]
Therefore
\[
\int_{\mathbb{R}^{N}}\frac{|\mathcal{R}_{1}u|^{2}}{\left\vert x\right\vert
}\mathrm{dx}=%
%TCIMACRO{\dsum \limits_{k=0}^{\infty}}%
%BeginExpansion
{\displaystyle\sum\limits_{k=0}^{\infty}}
%EndExpansion
\left(  \int_{0}^{\infty}r^{N+2k-2}\left\vert v_{k}^{\prime}\left(  r\right)
\right\vert ^{2}dr-\left(  c_{k}-k\right)  \int_{0}^{\infty}r^{N+2k-4}%
\left\vert v_{k}\left(  r\right)  \right\vert ^{2}dr\right)
\]

Finally%
\begin{align*}
&  \int_{\mathbb{R}^{N}}\left\vert \mathcal{R}_{2}u\right\vert ^{2}%
\mathrm{dx}\\
&  =\int_{0}^{\infty}\int_{\mathbb{S}^{N-1}}\left\vert
%TCIMACRO{\dsum \limits_{k=0}^{\infty}}%
%BeginExpansion
{\displaystyle\sum\limits_{k=0}^{\infty}}
%EndExpansion
\mathcal{R}_{2}\left(  u_{k}\left(  r\right)  \right)  \phi_{k}\left(
\sigma\right)  \right\vert ^{2}r^{N-1}d\sigma dr\\
&  =%
%TCIMACRO{\dsum \limits_{k=0}^{\infty}}%
%BeginExpansion
{\displaystyle\sum\limits_{k=0}^{\infty}}
%EndExpansion
\int_{0}^{\infty}\left\vert \left(  u_{k}^{\prime\prime}\left(  r\right)
+\frac{N-1}{r}u_{k}^{\prime}\left(  r\right)  \right)  \right\vert ^{2}%
r^{N-1}dr\\
&  =%
%TCIMACRO{\dsum \limits_{k=0}^{\infty}}%
%BeginExpansion
{\displaystyle\sum\limits_{k=0}^{\infty}}
%EndExpansion
\left(  \int_{0}^{\infty}\left\vert u_{k}^{\prime\prime}\left(  r\right)
\right\vert ^{2}r^{N-1}dr+\left(  N-1\right)  \int_{0}^{\infty}\left(
\left\vert u_{k}^{\prime}\left(  r\right)  \right\vert ^{2}\right)  ^{\prime
}r^{N-2}dr+\left(  N-1\right)  ^{2}\int_{0}^{\infty}\left\vert u_{k}^{\prime
}\left(  r\right)  \right\vert ^{2}r^{N-3}dr\right)
\end{align*}
For $N\geq3$, using the identity
\[
\int_{0}^{\infty}\left(  \left\vert u_{k}^{\prime}\left(  r\right)
\right\vert ^{2}\right)  ^{\prime}r^{N-2}dr=-\left(  N-2\right)  \int
_{0}^{\infty}\left\vert u_{k}^{\prime}\left(  r\right)  \right\vert
^{2}r^{N-3}dr,
\]
we obtain%
\[
\int_{\mathbb{R}^{N}}\left\vert \mathcal{R}_{2}u\right\vert ^{2}\mathrm{dx}=%
%TCIMACRO{\dsum \limits_{k=0}^{\infty}}%
%BeginExpansion
{\displaystyle\sum\limits_{k=0}^{\infty}}
%EndExpansion
\left(  \int_{0}^{\infty}\left\vert u_{k}^{\prime\prime}\left(  r\right)
\right\vert ^{2}r^{N-1}dr+\left(  N-1\right)  \int_{0}^{\infty}\left\vert
u_{k}^{\prime}\left(  r\right)  \right\vert ^{2}r^{N-3}dr\right)  .
\]
When $N=2$, we get
\[
\int_{\mathbb{R}^{2}}\left\vert \mathcal{R}_{2}u\right\vert ^{2}\mathrm{dx}=%
%TCIMACRO{\dsum \limits_{k=0}^{\infty}}%
%BeginExpansion
{\displaystyle\sum\limits_{k=0}^{\infty}}
%EndExpansion
\left(  \int_{0}^{\infty}\left\vert u_{k}^{\prime\prime}\left(  r\right)
\right\vert ^{2}rdr+\int_{0}^{\infty}\left(  \left\vert u_{k}^{\prime}\left(
r\right)  \right\vert ^{2}\right)  ^{\prime}dr+\int_{0}^{\infty}\left\vert
u_{k}^{\prime}\left(  r\right)  \right\vert ^{2}r^{-1}dr\right)
\]
Now, note that%
\begin{align*}
\int_{0}^{\infty}\left\vert u_{k}^{\prime}\left(  r\right)  \right\vert
^{2}r^{-1}dr  &  =\int_{0}^{\infty}\left\vert \left(  r^{k}v_{k}\left(
r\right)  \right)  ^{\prime}\right\vert ^{2}r^{-1}dr\\
&  =\int_{0}^{\infty}\left\vert r^{k}v_{k}^{\prime}\left(  r\right)
+kr^{k-1}v_{k}\left(  r\right)  \right\vert ^{2}r^{-1}dr\\
&  =\int_{0}^{\infty}\left\vert v_{k}^{\prime}\left(  r\right)  \right\vert
^{2}r^{2k-1}+2k\operatorname{Re}v_{k}^{\prime}\left(  r\right)  \overline
{v_{k}\left(  r\right)  }r^{2k-2}+k^{2}r^{2k-3}\left\vert v_{k}\left(
r\right)  \right\vert ^{2}dr.
\end{align*}
Therefore $\int_{0}^{\infty}\left\vert u_{k}^{\prime}\left(  r\right)
\right\vert ^{2}r^{-1}dr$ is well-defined for all $k\geq2$. For $k=0$, then
$\int_{0}^{\infty}\left\vert u_{0}^{\prime}\left(  r\right)  \right\vert
^{2}r^{-1}dr=\int_{0}^{\infty}\left\vert v_{0}^{\prime}\left(  r\right)
\right\vert ^{2}r^{-1}dr$ is well-defined from the proof of the Lemma
\ref{l3.1}. Therefore, since%
\[
\int_{\mathbb{R}^{2}}\left\vert \mathcal{R}_{2}u\right\vert ^{2}\mathrm{dx}=%
%TCIMACRO{\dsum \limits_{k=0}^{\infty}}%
%BeginExpansion
{\displaystyle\sum\limits_{k=0}^{\infty}}
%EndExpansion
\left(  \int_{0}^{\infty}\left\vert u_{k}^{\prime\prime}\left(  r\right)
\right\vert ^{2}rdr+\int_{0}^{\infty}\left(  \left\vert u_{k}^{\prime}\left(
r\right)  \right\vert ^{2}\right)  ^{\prime}dr+\int_{0}^{\infty}\left\vert
u_{k}^{\prime}\left(  r\right)  \right\vert ^{2}r^{-1}dr\right)  \text{,}%
\]
we deduce that $\int_{0}^{\infty}\left\vert u_{1}^{\prime}\left(  r\right)
\right\vert ^{2}r^{-1}dr$ must also be well-defined. That is $2\int
_{0}^{\infty}v_{1}^{\prime}\left(  r\right)  v_{1}\left(  r\right)
dr+\int_{0}^{\infty}\left\vert v_{1}\left(  r\right)  \right\vert ^{2}%
r^{-1}dr$ is well-defined. This implies $\int_{0}^{\infty}\left\vert
v_{1}\left(  r\right)  \right\vert ^{2}r^{-1}dr$ is well-defined and
$v_{1}\left(  0\right)  =0$.

Now, since
\[
u_{k}^{\prime}\left(  r\right)  =\left(  r^{k}v_{k}\left(  r\right)  \right)
^{\prime}=kr^{k-1}v_{k}\left(  r\right)  +r^{k}v_{k}^{\prime}\left(  r\right)
\]
and $v_{0}^{\prime}(0)=v_{1}\left(  0\right)  =0$, we deduce that
$u_{k}^{\prime}\left(  0\right)  =0$ for all $k\geq0$. Hence, we have%
\[
\int_{0}^{\infty}\left(  \left\vert u_{k}^{\prime}\left(  r\right)
\right\vert ^{2}\right)  ^{\prime}dr=0.
\]
Therefore
\[
\int_{\mathbb{R}^{2}}\left\vert \mathcal{R}_{2}u\right\vert ^{2}\mathrm{dx}=%
%TCIMACRO{\dsum \limits_{k=0}^{\infty}}%
%BeginExpansion
{\displaystyle\sum\limits_{k=0}^{\infty}}
%EndExpansion
\left(  \int_{0}^{\infty}\left\vert u_{k}^{\prime\prime}\left(  r\right)
\right\vert ^{2}rdr+\int_{0}^{\infty}\left\vert u_{k}^{\prime}\left(
r\right)  \right\vert ^{2}r^{-1}dr\right)  .
\]

\end{proof}

In computing sharp constants, we will also need to following lemmas:

\begin{lemma}
\label{l3.3}For any $N\geq2$
\begin{equation}
\inf_{k\in%
%TCIMACRO{\U{2115} }%
%BeginExpansion
\mathbb{N}
%EndExpansion
\cup\left\{  0\right\}  }\left(  1-\frac{8k}{(N+2k)^{2}}\right)
\frac{(N+2k+2)^{2}}{4}=\frac{(N+2)^{2}}{4}. \label{best_C}%
\end{equation}

\end{lemma}

\begin{proof}
Let
\[
S(N,k):=\left(  1-\frac{8k}{(N+2k)^{2}}\right)  \frac{(N+2k+2)^{2}}{4}.
\]
First we write $S(N,k)$ in a more convenient form, that is
\[
S(N,k)=\frac{N^{2}}{4}+N-3+Nk+k^{2}+\frac{4(N-1)}{N+2k}+\frac{4N}{(N+2k)^{2}%
}.
\]
Then we consider the extended function on the real positive axis
\[
S(N,x):=\frac{N^{2}}{4}+N-3+Nx+x^{2}+\frac{4(N-1)}{N+2x}+\frac{4N}{(N+2x)^{2}%
},\quad x\in\lbrack0,\infty).
\]
for which we study the monotonicity using basic differentiable techniques. By
simple computations we obtain
\[
\frac{\partial S}{\partial x}(N,x)=\frac{t^{4}-8(N-1)t-16N}{t^{3}}:=\frac{h(N,t)}{t^{3}}%
,\quad\text{ where }t:=2x+N,\quad t\geq N.
\]
Then $\frac{\partial h}{\partial t}(N,t)=4t^{3}-8(N-1)\geq4(N^{3}-2N+2)>0$, for any $t\geq
N\geq2$. So $h(N,\cdot)$ is a nondecreasing function with respect to $t$.

We distinguish the following cases.

\textbf{The case $N\geq4$}

From above we have $h(N,t)\geq h(N,N)=N^{4}-8N^{2}-8N=(N+2)\left(
(N-2)^{2}(N+2)-8\right)  >0$ for any $t\geq N\geq4$. Since $h$ is nonnegative
then $S$ is nondecreasing with respect to $x$ and we get that $S(N,x)\geq
S(N,0)$ for any $x\in\lbrack0,\infty)$ and $N\geq4$ which leads to
$\min_{k\geq0}S(N,k)=S(N,0)$ for any $N\geq4$.

\textbf{The case $N=2$}

We obtain
\[
\frac{\partial S}{\partial x}(2,x)=\frac{2(t^{4}-t-2)}{t^{3}}:=\frac{f(t)}{t^{3}},\quad\text{
where }t:=1+x,\quad x\geq0,\quad t\geq1.
\]
Since $f^{\prime}(t)=8t^{3}-2>0$ for any $t\geq1$ we have that for $t\geq2$,
$f(t)\geq f(2)=24>0$. Then $S(2,\cdot)$ is nondecreasing with respect to
$x\geq1$ and therefore,
\[
\inf_{k\geq0}S(2,k)=\min\left\{  S(2,0),S(2,1)\right\}  =S(2,0)=4.
\]

\textbf{The case $N=3$}

We obtain
\[
\frac{\partial S}{\partial x}(3,x)=\frac{t^{4}-16t-48}{t^{3}}:=\frac{g(t)}{t^{3}},\quad\text{
where }t:=3+2x,\quad x\geq0,\quad t\geq3.
\]
Then $g^{\prime}(t)=4(t^{3}-4)>0$ and so $g$ is nondecreasing on $[3,\infty)$.
In particular, $g(t)\geq g(5)>0$ for any $t\geq5$. Equivalently, $\partial
_{x}S(3,x)>0$ for any $x\geq1$ which implies $S(3,x)\geq S(3,1)=833/100$ for
any $x\geq1$. Since $S(3,0)=25/4$ we obtain that
\[
\min_{k\geq0}S(3,k)=S(3,0)=\frac{25}{4}\text{.}%
\]
The proof is completed now.
\end{proof}

\begin{lemma}
\label{l3.4}For any $N\geq5$
\begin{equation}
\inf_{k\in%
%TCIMACRO{\U{2115} }%
%BeginExpansion
\mathbb{N}
%EndExpansion
\cup\left\{  0\right\}  }\frac{\left(  N+2k+1\right)  ^{2}}{4}\frac
{(N+2k-3)^{4}}{\left(  (N+2k-3)^{2}+4k\right)  ^{2}}=\frac{(N+1)^{2}}{4}.
\end{equation}

\end{lemma}

\begin{proof}
We will first show that the function
\[
f\left(  x\right)  =\frac{\left(  N+2x+1\right)  ^{2}}{4}\frac{(N+2x-3)^{4}%
}{\left(  (N+2x-3)^{2}+4x\right)  ^{2}}%
\]
is increasing on $\left[  1,\infty\right)  $. Indeed, consider the function
$g:\left[  1,\infty\right)  \times\left[  4,\infty\right)  \rightarrow%
%TCIMACRO{\U{211d} }%
%BeginExpansion
\mathbb{R}
%EndExpansion
:$%
\[
g\left(  x,t\right)  =\frac{t^{4}\left(  t+4\right)  ^{2}}{4\left(
4x+t^{2}\right)  ^{2}}, \quad t:=N+2x-3.
\]
Then direct computation yields
\[
\frac{\partial g}{\partial x}\left(  x,t\right)  =-\frac{2t^{4}\left(
t+4\right)  ^{2}}{\left(  4x+t^{2}\right)  ^{3}}%
\]
and
\[
\frac{\partial g}{\partial t}\left(  x,t\right)  =\frac{\left(  t+4\right)
\left(  4x\left(  3t+8\right)  t^{3}+t^{6}\right)  }{2\left(  4x+t^{2}\right)
^{3}}.
\]
Note that $f\left(  x\right)  =g\left(  x,t\right)  $ with $t=N+2x-3$, we get
with $x\geq1$ and $t\geq4$ that%
\begin{align*}
f^{\prime}\left(  x\right)   &  =\frac{\partial g}{\partial x}\left(
x,t\right)  +2\frac{\partial g}{\partial t}\left(  x,t\right)  \\
&  =\frac{\left(  t+4\right)  \left(  4x\left(  3t+8\right)  t^{3}%
+t^{6}\right)  -2t^{4}\left(  t+4\right)  ^{2}}{\left(  4x+t^{2}\right)  ^{3}%
}\\
&  \geq\frac{\left(  t+4\right)  t^{6}-2t^{4}\left(  t+4\right)  ^{2}}{\left(
4x+t^{2}\right)  ^{3}}\\
&  \geq\frac{\left(  t+4\right)  t^{4}}{\left(  4x+t^{2}\right)  ^{3}}\left(
t-4\right)  \left(  t+2\right)  \\
&  \geq0.
\end{align*}
Therefore
\begin{align*}
&  \inf_{k\in%
%TCIMACRO{\U{2115} }%
%BeginExpansion
\mathbb{N}
%EndExpansion
\cup\left\{  0\right\}  }\frac{\left(  N+2k+1\right)  ^{2}}{4}\frac
{(N+2k-3)^{4}}{\left(  (N+2k-3)^{2}+4k\right)  ^{2}}\\
&  =\min_{k=0,1}\frac{\left(  N+2k+1\right)  ^{2}}{4}\frac{(N+2k-3)^{4}%
}{\left(  (N+2k-3)^{2}+4k\right)  ^{2}}\\
&  =\min\left\{  \frac{(N+1)^{2}}{4},~\frac{\left(  N+3\right)  ^{2}}{4}%
\frac{(N-1)^{4}}{\left(  (N-1)^{2}+4\right)  ^{2}}\right\}  .
\end{align*}
We now will prove that for $N\geq5$,
\[
\frac{\left(  N+3\right)  ^{2}}{4}\frac{(N-1)^{4}}{\left(  (N-1)^{2}+4\right)
^{2}}\geq\frac{(N+1)^{2}}{4}.
\]
Let $x=N-1\geq4$. The above inequality is equivalent to
\[
x^{2}\left(  x+4\right)  \geq\left(  x^{2}+4\right)  \left(  x+2\right)
\]
or
\[
x^{2}\geq2x+4
\]
which is true for $x\geq4$.
\end{proof}

\section{Proofs of main results}

\subsection{Proof of Theorem \ref{th1}}

\paragraph{\textbf{The case }$N=1$}

We need to prove that for $u\in C_{c}^{\infty}(\mathbb{-\infty},\mathbb{\infty
}):$%
\[
\int_{-\infty}^{\infty}|u^{\prime\prime}(r)|^{2}dr\int_{-\infty}^{\infty}%
r^{2}|u^{\prime}(r)|^{2}dr\geq\frac{9}{4}\left(  \int_{-\infty}^{\infty
}|u^{\prime}(r)|^{2}dr\right)  ^{2}.
\]
Indeed, for any $x\in%
%TCIMACRO{\U{211d} }%
%BeginExpansion
\mathbb{R}
%EndExpansion
$, we have%
\begin{align*}
&  \int_{-\infty}^{\infty}\left\vert u^{\prime\prime}\left(  r\right)
+xru^{\prime}\left(  r\right)  +xu\left(  r\right)  \right\vert ^{2}dr\\
&  =\int_{-\infty}^{\infty}\left\vert u^{\prime\prime}\left(  r\right)
\right\vert ^{2}dr+x^{2}\int_{-\infty}^{\infty}r^{2}\left\vert u^{\prime
}\left(  r\right)  \right\vert ^{2}dr+x^{2}\int_{-\infty}^{\infty}\left\vert
u\left(  r\right)  \right\vert ^{2}dr\\
&  +2x\operatorname{Re}\int_{-\infty}^{\infty}\overline{u^{\prime\prime
}\left(  r\right)  }ru^{\prime}\left(  r\right)  dr+2x\operatorname{Re}%
\int_{-\infty}^{\infty}\overline{u^{\prime\prime}\left(  r\right)  }u\left(
r\right)  dr+2x^{2}\operatorname{Re}\int_{-\infty}^{\infty}r\overline
{u^{\prime}\left(  r\right)  }u\left(  r\right)  dr\\
&  =\int_{-\infty}^{\infty}\left\vert u^{\prime\prime}\left(  r\right)
\right\vert ^{2}dr+x^{2}\int_{-\infty}^{\infty}r^{2}\left\vert u^{\prime
}\left(  r\right)  \right\vert ^{2}dr+x^{2}\int_{-\infty}^{\infty}\left\vert
u\left(  r\right)  \right\vert ^{2}dr\\
&  -3x\int_{-\infty}^{\infty}\left\vert u^{\prime}\left(  r\right)
\right\vert ^{2}dr-x^{2}\int_{-\infty}^{\infty}\left\vert u\left(  r\right)
\right\vert ^{2}dr\\
&  =\left(  \int_{-\infty}^{\infty}r^{2}\left\vert u^{\prime}\left(  r\right)
\right\vert ^{2}dr\right)  x^{2}-3\left(  \int_{-\infty}^{\infty}\left\vert
u^{\prime}\left(  r\right)  \right\vert ^{2}dr\right)  x+\int_{-\infty
}^{\infty}\left\vert u^{\prime\prime}\left(  r\right)  \right\vert ^{2}dr.
\end{align*}
Since $\int_{-\infty}^{\infty}\left\vert u^{\prime\prime}\left(  r\right)
+xru^{\prime}\left(  r\right)  +xu\left(  r\right)  \right\vert ^{2}dr\geq0$
for all $x\in%
%TCIMACRO{\U{211d} }%
%BeginExpansion
\mathbb{R}
%EndExpansion
$, we deduce that
\[
\left[  3\left(  \int_{-\infty}^{\infty}\left\vert u^{\prime}\left(  r\right)
\right\vert ^{2}dr\right)  \right]  ^{2}\leq4\left(  \int_{-\infty}^{\infty
}r^{2}\left\vert u^{\prime}\left(  r\right)  \right\vert ^{2}dr\right)
\left(  \int_{-\infty}^{\infty}\left\vert u^{\prime\prime}\left(  r\right)
\right\vert ^{2}dr\right)  .
\]
Equivalently,
\[
\int_{-\infty}^{\infty}|u^{\prime\prime}(r)|^{2}dr\int_{-\infty}^{\infty}%
r^{2}|u^{\prime}(r)|^{2}dr\geq\frac{9}{4}\left(  \int_{-\infty}^{\infty
}|u^{\prime}(r)|^{2}dr\right)  ^{2}.
\]

\paragraph{\textbf{The case }$N\geq2$}

Let $u\in C_{c}^{\infty}(\mathbb{R}^{N})$. By Lemma \ref{l3.1}, inequality
\eqref{main1} is equivalent to
\begin{multline}
\sum_{k=0}^{\infty}\left(  \int_{0}^{\infty}r^{N+2k-1}|v_{k}^{\prime\prime
}|^{2}dr+(N-1+2k)\int_{0}^{\infty}r^{N+2k-3}|v_{k}^{\prime}|^{2}dr\right)
\label{to_be_proven}\\
\times\sum_{k=0}^{\infty}\left(  \int_{0}^{\infty}r^{N+2k+1}|v_{k}^{\prime
}|^{2}dr-2k\int_{0}^{\infty}r^{N+2k-1}\left\vert v_{k}\right\vert
^{2}dr\right)  \\
\geq\mu^{\star\ast}(N)\left(  \sum_{k=0}^{\infty}\int_{0}^{\infty}%
r^{N+2k-1}|v_{k}^{\prime}|^{2}dr\right)  ^{2}%
\end{multline}
Due to Cauchy-Bunyakovsky-Schwarz inequality in order to justify
\eqref{to_be_proven} it is enough to show for any $k\geq0$ that
\begin{multline}
\left(  \int_{0}^{\infty}r^{N+2k-1}|v_{k}^{\prime\prime}|^{2}dr+(N-1+2k)\int
_{0}^{\infty}r^{N+2k-3}|v_{k}^{\prime}|^{2}dr\right)  \label{to_be_proven3}\\
\times\left(  \int_{0}^{\infty}r^{N+2k+1}|v_{k}^{\prime}|^{2}dr-2k\int
_{0}^{\infty}r^{N+2k-1}\left\vert v_{k}\right\vert ^{2}dr\right)  \\
\geq\mu^{\star\ast}(N)\left(  \int_{0}^{\infty}r^{N+2k-1}|v_{k}^{\prime}%
|^{2}dr\right)  ^{2}%
\end{multline}
For $k\geq1$, we use the following weighted 1-d Hardy inequality:
\begin{equation}
\int_{0}^{\infty}r^{N+2k+1}|v_{k}^{\prime}|^{2}dr\geq\frac{(N+2k)^{2}}{4}%
\int_{0}^{\infty}r^{N+2k-1}\left\vert v_{k}\right\vert ^{2}%
dr\label{weighted_1d_H}%
\end{equation}
Therefore, it is now enough to show that
\begin{align}
&  \left(  1-\frac{8k}{(N+2k)^{2}}\right)  \left(  \int_{0}^{\infty}%
r^{N+2k-1}|v_{k}^{\prime\prime}|^{2}dr+(N-1+2k)\int_{0}^{\infty}%
r^{N+2k-3}|v_{k}^{\prime}|^{2}dr\right)  \nonumber\label{ineq_1}\\
&  \times\left(  \int_{0}^{\infty}r^{N+2k+1}|v_{k}^{\prime}|^{2}dr\right)
\geq\mu^{\star\ast}(N)\left(  \int_{0}^{\infty}r^{N+2k-1}|v_{k}^{\prime}%
|^{2}dr\right)  ^{2},\quad\forall k\geq0.
\end{align}
For that we will try to get the best constant $C_{N,k}$ in the inequality
\begin{align}
&  \left(  \int_{0}^{\infty}r^{N+2k-1}|v_{k}^{\prime\prime}|^{2}%
dr+(N-1+2k)\int_{0}^{\infty}r^{N+2k-3}|v_{k}^{\prime}|^{2}dr\right)  \left(
\int_{0}^{\infty}r^{N+2k+1}|v_{k}^{\prime}|^{2}dr\right)  \nonumber\\
&  \geq C_{N,k}\left(  \int_{0}^{\infty}r^{N+2k-1}|v_{k}^{\prime}%
|^{2}dr\right)  ^{2},\quad\forall k\geq0.\label{4.0}%
\end{align}
Then%
\[
\mu^{\star\ast}(N)\geq\inf_{k\in%
%TCIMACRO{\U{2115} }%
%BeginExpansion
\mathbb{N}
%EndExpansion
\cup\left\{  0\right\}  }\left(  1-\frac{8k}{(N+2k)^{2}}\right)  C_{N,k}.
\]
Denoting $w=v_{k}^{\prime}$, \eqref{4.0} reduces to%

\begin{align}
&  \left(  \int_{0}^{\infty}r^{N+2k-1}|w^{\prime}|^{2}dr+(N-1+2k)\int
_{0}^{\infty}r^{N+2k-3}|w|^{2}dr\right)  \left(  \int_{0}^{\infty}%
r^{N+2k+1}|w|^{2}dr\right) \nonumber\label{ineq_2}\\
&  \geq C_{N,k}\left(  \int_{0}^{\infty}r^{N+2k-1}|w|^{2}dr\right)  ^{2}%
,\quad\forall k\geq0.
\end{align}
Fix $\varepsilon\gtrapprox0$. Let $w\left(  r\right)  =rv\left(  r\right)  $
on $\left[  \varepsilon,\infty\right)  $. Then for $r\in\left(  \varepsilon
,\infty\right)  $, we have
\[
w^{\prime}\left(  r\right)  =v\left(  r\right)  +rv^{\prime}\left(  r\right)
.
\]
Therefore,%
\begin{align*}
&  r^{N+2k-1}\left\vert w^{\prime}\left(  r\right)  \right\vert ^{2}%
+(N-1+2k)r^{N+2k-3}\left\vert w\left(  r\right)  \right\vert ^{2}\\
&  =(N+2k)r^{N+2k-1}v^{2}\left(  r\right)  +r^{N+2k+1}\left\vert v^{\prime
}\left(  r\right)  \right\vert ^{2}+2r^{N+2k}v\left(  r\right)  v^{\prime
}\left(  r\right)
\end{align*}
and
\begin{align*}
&  \int_{\varepsilon}^{\infty}\left(  r^{N+2k-1}\left\vert w^{\prime}\left(
r\right)  \right\vert ^{2}+(N-1+2k)r^{N+2k-3}\left\vert w\left(  r\right)
\right\vert ^{2}\right)  dr\\
&  =\int_{\varepsilon}^{\infty}\left(  (N+2k)r^{N+2k-1}\left\vert v\left(
r\right)  \right\vert ^{2}+r^{N+2k+1}\left\vert v^{\prime}\left(  r\right)
\right\vert ^{2}+2\operatorname{Re}r^{N+2k}\overline{v\left(  r\right)
}v^{\prime}\left(  r\right)  \right)  dr.
\end{align*}
By noting that
\begin{align*}
\operatorname{Re}\int_{\varepsilon}^{\infty}2r^{N+2k}\overline{v\left(
r\right)  }v^{\prime}\left(  r\right)   &  =\int_{\varepsilon}^{\infty
}r^{N+2k}\left\vert v^{2}\left(  r\right)  \right\vert ^{\prime}dr\\
&  =o_{\varepsilon}\left(  1\right)  -(N+2k)\int_{\varepsilon}^{\infty
}r^{N+2k-1}\left\vert v\left(  r\right)  \right\vert ^{2}dr,
\end{align*}
we get%
\[
\int_{\varepsilon}^{\infty}\left(  r^{N+2k-1}\left\vert w^{\prime}\left(
r\right)  \right\vert ^{2}+(N-1+2k)r^{N+2k-3}\left\vert w\left(  r\right)
\right\vert ^{2}\right)  dr=\int_{\varepsilon}^{\infty}r^{N+2k+1}\left\vert
v^{\prime}\left(  r\right)  \right\vert ^{2}dr+o_{\varepsilon}\left(
1\right)  .
\]
Hence, using integration by parts and H\"{o}lder inequality, we get
\begin{align*}
&  \left(  \int_{\varepsilon}^{\infty}\left(  r^{N+2k-1}\left\vert w^{\prime
}\left(  r\right)  \right\vert ^{2}+(N-1+2k)r^{N+2k-3}\left\vert w\left(
r\right)  \right\vert ^{2}\right)  dr\right)  \left(  \int_{\varepsilon
}^{\infty}r^{N+2k+1}|w|^{2}dr\right) \\
&  =\left(  \int_{\varepsilon}^{\infty}r^{N+2k+1}\left\vert v^{\prime}\left(
r\right)  \right\vert ^{2}dr+o_{\varepsilon}\left(  1\right)  \right)  \left(
\int_{\varepsilon}^{\infty}r^{N+2k+3}|v(r)|^{2}dr\right) \\
&  =\left(  \int_{\varepsilon}^{\infty}r^{N+2k+1}\left\vert v^{\prime}\left(
r\right)  \right\vert ^{2}dr\right)  \left(  \int_{\varepsilon}^{\infty
}r^{N+2k+3}|v(r)|^{2}dr\right)  +o_{\varepsilon}\left(  1\right) \\
&  \geq\left(  \int_{\varepsilon}^{\infty}r^{N+2k+2}\left\vert v(r)v^{\prime
}(r)\right\vert dr\right)  ^{2}+o_{\varepsilon}\left(  1\right) \\
&  =\frac{1}{4}\left(  \int_{\varepsilon}^{\infty}r^{N+2k+2}\left(
|v(r)|^{2}\right)  ^{\prime}dr\right)  ^{2}+o_{\varepsilon}\left(  1\right) \\
&  =\frac{\left(  N+2k+2\right)  ^{2}}{4}\left(  \int_{\varepsilon}^{\infty
}r^{N+2k+1}|v(r)|^{2}dr\right)  ^{2}+o_{\varepsilon}\left(  1\right) \\
&  =\frac{\left(  N+2k+2\right)  ^{2}}{4}\left(  \int_{\varepsilon}^{\infty
}r^{N+2k-1}|w(r)|^{2}dr\right)  ^{2}+o_{\varepsilon}\left(  1\right)  .
\end{align*}
Letting $\varepsilon\downarrow0$, we obtain%
\begin{align}
&  \left(  \int_{0}^{\infty}r^{N+2k-1}|w^{\prime}|^{2}dr+(N-1+2k)\int
_{0}^{\infty}r^{N+2k-3}|w|^{2}dr\right)  \left(  \int_{0}^{\infty}%
r^{N+2k+1}|w|^{2}dr\right) \nonumber\\
&  \geq\frac{\left(  N+2k+2\right)  ^{2}}{4}\left(  \int_{0}^{\infty
}r^{N+2k-1}|w|^{2}dr\right)  ^{2},\quad\forall k\geq0. \label{4.1}%
\end{align}
Therefore, by Lemma \ref{l3.1},%
\[
\mu^{\star\ast}(N)\geq\inf_{k\in%
%TCIMACRO{\U{2115} }%
%BeginExpansion
\mathbb{N}
%EndExpansion
\cup\left\{  0\right\}  }\left(  1-\frac{8k}{(N+2k)^{2}}\right)  \frac{\left(
N+2k+2\right)  ^{2}}{4}=\frac{(N+2)^{2}}{4}.
\]
In other words, for all $u\in C_{c}^{\infty}(\mathbb{R}^{N})$, we have%
\[
\int_{\mathbb{R}^{N}}|\Delta u|^{2}\mathrm{dx}\int_{\mathbb{R}^{N}}%
|x|^{2}|\nabla u|^{2}\mathrm{dx}\geq\frac{(N+2)^{2}}{4}\left(  \int
_{\mathbb{R}^{N}}|\nabla u|^{2}\mathrm{dx}\right)  ^{2}.
\]
Hence, by standard density argument, for all $u\in\mathcal{S}(\mathbb{R}^{N}%
)$, we have%
\[
\int_{\mathbb{R}^{N}}|\Delta u|^{2}\mathrm{dx}\int_{\mathbb{R}^{N}}%
|x|^{2}|\nabla u|^{2}\mathrm{dx}\geq\frac{(N+2)^{2}}{4}\left(  \int
_{\mathbb{R}^{N}}|\nabla u|^{2}\mathrm{dx}\right)  ^{2}.
\]

\paragraph{\textbf{Attainability of the best constant }$\frac{(N+2)^{2}}{4}$}

Now, we show that \eqref{main1} is attained by Gaussian profiles of the form
$u(x)=\alpha e^{-\beta|x|^{2}}$, $\beta>0$, $\alpha\in\mathbb{%
%TCIMACRO{\U{2102} }%
%BeginExpansion
\mathbb{C}
%EndExpansion
}$. Indeed, direct computations yield%
\begin{align*}
&  \int_{\mathbb{R}^{N}}|\Delta u|^{2}\mathrm{dx}\\
&  =\left\vert \mathbb{S}^{N-1}\right\vert \int_{0}^{\infty}\left\vert
u^{\prime\prime}\left(  r\right)  +\frac{N-1}{r}u^{\prime}\left(  r\right)
\right\vert ^{2}r^{N-1}dr\\
&  =\left\vert \mathbb{S}^{N-1}\right\vert \int_{0}^{\infty}\left\vert
u^{\prime\prime}\left(  r\right)  +\frac{N-1}{r}u^{\prime}\left(  r\right)
\right\vert ^{2}r^{N-1}dr\\
&  =4\left\vert \alpha\right\vert ^{2}\beta^{2}\left\vert \mathbb{S}%
^{N-1}\right\vert \int_{0}^{\infty}\left\vert \left(  2\beta r^{2}-N\right)
\right\vert ^{2}e^{-2\beta r^{2}}r^{N-1}dr\\
&  =4\left\vert \alpha\right\vert ^{2}\beta^{2}\left\vert \mathbb{S}%
^{N-1}\right\vert \left[  4\beta^{2}\int_{0}^{\infty}e^{-2\beta r^{2}}%
r^{N+3}dr-4\beta N\int_{0}^{\infty}e^{-2\beta r^{2}}r^{N+1}dr+N^{2}\int
_{0}^{\infty}e^{-2\beta r^{2}}r^{N-1}dr\right]  .
\end{align*}
Note that
\[
\int_{0}^{\infty}e^{-2\beta r^{2}}r^{N-1}dr=\frac{1}{\left(  2\beta\right)
^{\frac{N}{2}}}\int_{0}^{\infty}e^{-t^{2}}t^{N-1}dt=\frac{1}{\left(
2\beta\right)  ^{\frac{N}{2}}}\frac{\Gamma\left(  \frac{N}{2}\right)  }{2},
\]%
\[
\int_{0}^{\infty}e^{-2\beta r^{2}}r^{N+1}dr=\frac{1}{\left(  2\beta\right)
^{\frac{N}{2}+1}}\int_{0}^{\infty}e^{-t^{2}}t^{N+1}dt=\frac{1}{\left(
2\beta\right)  ^{\frac{N}{2}+1}}\frac{\Gamma\left(  \frac{N}{2}+1\right)  }%
{2},
\]
and%
\[
\int_{0}^{\infty}e^{-2\beta r^{2}}r^{N+3}dr=\frac{1}{\left(  2\beta\right)
^{\frac{N}{2}+2}}\int_{0}^{\infty}e^{-t^{2}}t^{N+3}dt=\frac{1}{\left(
2\beta\right)  ^{\frac{N}{2}+2}}\frac{\Gamma\left(  \frac{N}{2}+2\right)  }%
{2}.
\]
Hence%
\[
\int_{\mathbb{R}^{N}}|\Delta u|^{2}\mathrm{dx}=4\left\vert \alpha\right\vert
^{2}\beta^{2}\left\vert \mathbb{S}^{N-1}\right\vert \frac{1}{\left(
2\beta\right)  ^{\frac{N}{2}}}\frac{N\left(  N+2\right)  }{4}\frac
{\Gamma\left(  \frac{N}{2}\right)  }{2}.
\]
Also%
\begin{align*}
\int_{\mathbb{R}^{N}}|x|^{2}|\nabla u|^{2}\mathrm{dx} &  =\left\vert
\mathbb{S}^{N-1}\right\vert \int_{0}^{\infty}\left\vert u^{\prime}\left(
r\right)  \right\vert ^{2}r^{N+1}dr\\
&  =4\left\vert \alpha\right\vert ^{2}\beta^{2}\left\vert \mathbb{S}%
^{N-1}\right\vert \int_{0}^{\infty}e^{-2\beta r^{2}}r^{N+3}dr\\
&  =4\left\vert \alpha\right\vert ^{2}\beta^{2}\left\vert \mathbb{S}%
^{N-1}\right\vert \frac{1}{\left(  2\beta\right)  ^{\frac{N}{2}}}%
\frac{N\left(  N+2\right)  }{16\beta^{2}}\frac{\Gamma\left(  \frac{N}%
{2}\right)  }{2}%
\end{align*}
and%
\begin{align*}
\int_{\mathbb{R}^{N}}|\nabla u|^{2}\mathrm{dx} &  =\left\vert \mathbb{S}%
^{N-1}\right\vert \int_{0}^{\infty}\left\vert u^{\prime}\left(  r\right)
\right\vert ^{2}r^{N-1}dr\\
&  =4\left\vert \alpha\right\vert ^{2}\beta^{2}\left\vert \mathbb{S}%
^{N-1}\right\vert \int_{0}^{\infty}e^{-2\beta r^{2}}r^{N+1}dr\\
&  =4\left\vert \alpha\right\vert ^{2}\beta^{2}\left\vert \mathbb{S}%
^{N-1}\right\vert \frac{1}{\left(  2\beta\right)  ^{\frac{N}{2}}}\frac
{N}{4\beta}\frac{\Gamma\left(  \frac{N}{2}\right)  }{2}%
\end{align*}
Therefore%
\[
\frac{\int_{\mathbb{R}^{N}}|\Delta u|^{2}\mathrm{dx}\int_{\mathbb{R}^{N}%
}|x|^{2}|\nabla u|^{2}\mathrm{dx}}{\left(  \int_{\mathbb{R}^{N}}|\nabla
u|^{2}\mathrm{dx}\right)  ^{2}}=\frac{\frac{N\left(  N+2\right)  }{4}%
\frac{N\left(  N+2\right)  }{16\beta^{2}}}{\left(  \frac{N}{4\beta}\right)
^{2}}=\frac{(N+2)^{2}}{4}\text{.}%
\]

\subsection{Proof of Theorem \ref{th2}}

\paragraph{\textbf{The case }$N=1$}

In this case, it is obvious that $\left\vert \mathcal{R}_{2}u\right\vert
=\left\vert \Delta u\right\vert =\left\vert u^{\prime\prime}\right\vert $ and
$\left\vert \mathcal{R}_{1}u\right\vert =\left\vert \nabla u\right\vert
=\left\vert u^{\prime}\right\vert $. Therefore \eqref{main2} is a direct
consequence of \eqref{main1}.

\paragraph{\textbf{The case }$N\geq2$}

Let $u\in C_{c}^{\infty}(\mathbb{R}^{N})$. By Lemma \ref{l3.2}, inequality
\eqref{main2} is equivalent to
\begin{multline}
\sum_{k=0}^{\infty}\left(  \int_{0}^{\infty}\left\vert u_{k}^{\prime\prime
}\left(  r\right)  \right\vert ^{2}r^{N-1}dr+\left(  N-1\right)  \int
_{0}^{\infty}\left\vert u_{k}^{\prime}\left(  r\right)  \right\vert
^{2}r^{N-3}dr\right)  \\
\times\sum_{k=0}^{\infty}\int_{0}^{\infty}\left\vert u_{k}^{\prime}\left(
r\right)  \right\vert ^{2}r^{N+1}dr\\
\geq\frac{(N+2)^{2}}{4}\left(
%TCIMACRO{\dsum \limits_{k=0}^{\infty}}%
%BeginExpansion
{\displaystyle\sum\limits_{k=0}^{\infty}}
%EndExpansion
\int_{0}^{\infty}\left\vert u_{k}^{\prime}\left(  r\right)  \right\vert
^{2}r^{N-1}dr\right)  ^{2}.
\end{multline}
Therefore, it is enough to show that for all $v\in C_{c}^{\infty}%
(\mathbb{R}^{N}):$
\begin{multline}
\left(  \int_{0}^{\infty}\left\vert v^{\prime\prime}\left(  r\right)
\right\vert ^{2}r^{N-1}dr+\left(  N-1\right)  \int_{0}^{\infty}\left\vert
v^{\prime}\left(  r\right)  \right\vert ^{2}r^{N-3}dr\right)  \\
\times\left(  \int_{0}^{\infty}\left\vert v^{\prime}\left(  r\right)
\right\vert ^{2}r^{N+1}dr\right)  \\
\geq\frac{(N+2)^{2}}{4}\left(  \int_{0}^{\infty}\left\vert v^{\prime}\left(
r\right)  \right\vert ^{2}r^{N-1}dr\right)  ^{2}.
\end{multline}
But this is just (\ref{4.0}) with $k=0$. Therefore, by density argument,
\eqref{main2} also holds for $u\in\mathcal{S}(\mathbb{R}^{N})$.

Also, it is easy to check that with $u\left(  x\right)  =\alpha e^{-\beta
|x|^{2}}$, $\beta>0$, $\alpha\in\mathbb{%
%TCIMACRO{\U{2102} }%
%BeginExpansion
\mathbb{C}
%EndExpansion
}$, then
\[
\int_{\mathbb{R}^{N}}|\mathcal{R}_{2}u|^{2}\mathrm{dx}\int_{\mathbb{R}^{N}%
}|x|^{2}|\mathcal{R}_{1}u|^{2}\mathrm{dx}=\frac{(N+2)^{2}}{4}\left(
\int_{\mathbb{R}^{N}}|\mathcal{R}_{1}u|^{2}\mathrm{dx}\right)  ^{2}.
\]

\subsection{Proof of Theorem \ref{th3}}

Let $u\in C_{c}^{\infty}(\mathbb{R}^{N})$. By Lemma \ref{l3.1}, inequality
\eqref{main4} is equivalent to
\begin{multline}
\sum_{k=0}^{\infty}\left(  \int_{0}^{\infty}r^{N+2k-1}|v_{k}^{\prime\prime
}|^{2}dr+(N-1+2k)\int_{0}^{\infty}r^{N+2k-3}|v_{k}^{\prime}|^{2}dr\right)  \\
\times\sum_{k=0}^{\infty}\left(  \int_{0}^{\infty}r^{N+2k-1}|v_{k}^{\prime
}|^{2}dr\right)  \\
\geq\nu^{\star\ast}(N)\left(  \sum_{k=0}^{\infty}\int_{0}^{\infty}%
r^{N+2k-2}|v_{k}^{\prime}|^{2}dr+k\int_{0}^{\infty}r^{N+2k-4}\left\vert
v_{k}\right\vert ^{2}dr\right)  ^{2}.
\end{multline}
Due to Cauchy-Bunyakovsky-Schwarz inequality in order to justify
\eqref{main4}, it is enough to show for any $k\geq0$ that
\begin{multline}
\left(  \int_{0}^{\infty}r^{N+2k-1}|v_{k}^{\prime\prime}|^{2}dr+(N-1+2k)\int
_{0}^{\infty}r^{N+2k-3}|v_{k}^{\prime}|^{2}dr\right)  \\
\times\left(  \int_{0}^{\infty}r^{N+2k-1}|v_{k}^{\prime}|^{2}dr\right)  \\
\geq\nu^{\star\ast}(N)\left(  \int_{0}^{\infty}r^{N+2k-2}|v_{k}^{\prime}%
|^{2}dr+k\int_{0}^{\infty}r^{N+2k-4}\left\vert v_{k}\right\vert ^{2}dr\right)
^{2}.
\end{multline}
For $k\geq1$, we use the following weighted 1-d Hardy inequality
\begin{equation}
\int_{0}^{\infty}r^{N+2k-2}|v_{k}^{\prime}|^{2}dr\geq\frac{(N+2k-3)^{2}}%
{4}\int_{0}^{\infty}r^{N+2k-4}\left\vert v_{k}\right\vert ^{2}dr.
\end{equation}
Therefore, it is now enough to prove that for all $k\geq0:$
\begin{align}
&  \left(  \int_{0}^{\infty}r^{N+2k-1}|v_{k}^{\prime\prime}|^{2}%
dr+(N-1+2k)\int_{0}^{\infty}r^{N+2k-3}|v_{k}^{\prime}|^{2}dr\right)
\nonumber\\
&  \times\left(  \int_{0}^{\infty}r^{N+2k-1}|v_{k}^{\prime}|^{2}dr\right)
\geq\nu^{\star\ast}(N)\left(  1+\frac{4k}{(N+2k-3)^{2}}\right)  ^{2}\left(
\int_{0}^{\infty}r^{N+2k-2}|v_{k}^{\prime}|^{2}dr\right)  ^{2}\text{.}%
\end{align}
For that we will try to get the best constant $C_{N,k}$ in the inequality
\begin{align}
&  \left(  \int_{0}^{\infty}r^{N+2k-1}|v_{k}^{\prime\prime}|^{2}%
dr+(N-1+2k)\int_{0}^{\infty}r^{N+2k-3}|v_{k}^{\prime}|^{2}dr\right)  \left(
\int_{0}^{\infty}r^{N+2k-1}|v_{k}^{\prime}|^{2}dr\right)  \nonumber\\
&  \geq C_{N,k}\left(  \int_{0}^{\infty}r^{N+2k-2}|v_{k}^{\prime}%
|^{2}dr\right)  ^{2},\quad\forall k\geq0.\label{ineq_key_1}%
\end{align}
Then
\[
\nu^{\star\ast}(N)\geq\inf_{k\geq0}\frac{C_{N,k}}{\left(  1+\frac
{4k}{(N+2k-3)^{2}}\right)  ^{2}}.
\]
Denoting $w=v_{k}^{\prime}$, \eqref{ineq_key_1} reduces to%

\begin{align}
&  \left(  \int_{0}^{\infty}r^{N+2k-1}|w^{\prime}|^{2}dr+(N-1+2k)\int
_{0}^{\infty}r^{N+2k-3}|w|^{2}dr\right)  \left(  \int_{0}^{\infty}%
r^{N+2k-1}|w|^{2}dr\right)  \nonumber\\
&  \geq C_{N,k}\left(  \int_{0}^{\infty}r^{N+2k-2}|w|^{2}dr\right)  ^{2}%
,\quad\forall k\geq0.
\end{align}
Fix $\varepsilon\gtrapprox0$. Let $w\left(  r\right)  =rv\left(  r\right)  $
on $\left[  \varepsilon,\infty\right)  $. Then for $r\in\left(  \varepsilon
,\infty\right)  $, we have
\[
\int_{\varepsilon}^{\infty}\left(  r^{N+2k-1}\left\vert w^{\prime}\left(
r\right)  \right\vert ^{2}+(N-1+2k)r^{N+2k-3}\left\vert w\left(  r\right)
\right\vert ^{2}\right)  dr=\int_{\varepsilon}^{\infty}r^{N+2k+1}\left\vert
v^{\prime}\left(  r\right)  \right\vert ^{2}dr+o_{\varepsilon}\left(
1\right)  .
\]
Hence, using integration by parts and H\"{o}lder inequality, we get
\begin{align*}
&  \left(  \int_{\varepsilon}^{\infty}\left(  r^{N+2k-1}\left\vert w^{\prime
}\left(  r\right)  \right\vert ^{2}+(N-1+2k)r^{N+2k-3}\left\vert w\left(
r\right)  \right\vert ^{2}\right)  dr\right)  \left(  \int_{\varepsilon
}^{\infty}r^{N+2k-1}|w|^{2}dr\right)  \\
&  =\left(  \int_{\varepsilon}^{\infty}r^{N+2k+1}\left\vert v^{\prime}\left(
r\right)  \right\vert ^{2}dr+o_{\varepsilon}\left(  1\right)  \right)  \left(
\int_{\varepsilon}^{\infty}r^{N+2k+1}|v(r)|^{2}dr\right)  \\
&  =\left(  \int_{\varepsilon}^{\infty}r^{N+2k+1}\left\vert v^{\prime}\left(
r\right)  \right\vert ^{2}dr\right)  \left(  \int_{\varepsilon}^{\infty
}r^{N+2k+1}|v(r)|^{2}dr\right)  +o_{\varepsilon}\left(  1\right)  \\
&  \geq\left(  \int_{\varepsilon}^{\infty}r^{N+2k+1}\left\vert v(r)v^{\prime
}(r)\right\vert dr\right)  ^{2}+o_{\varepsilon}\left(  1\right)  \\
&  =\frac{1}{4}\left(  \int_{\varepsilon}^{\infty}r^{N+2k+1}\left(
|v(r)|^{2}\right)  ^{\prime}dr\right)  ^{2}+o_{\varepsilon}\left(  1\right)
\\
&  =\frac{\left(  N+2k+1\right)  ^{2}}{4}\left(  \int_{\varepsilon}^{\infty
}r^{N+2k}|v(r)|^{2}dr\right)  ^{2}+o_{\varepsilon}\left(  1\right)  \\
&  =\frac{\left(  N+2k+1\right)  ^{2}}{4}\left(  \int_{\varepsilon}^{\infty
}r^{N+2k-2}|w(r)|^{2}dr\right)  ^{2}+o_{\varepsilon}\left(  1\right)  .
\end{align*}
Letting $\varepsilon\downarrow0$, we obtain%
\begin{align}
&  \left(  \int_{0}^{\infty}r^{N+2k-1}|w^{\prime}|^{2}dr+(N-1+2k)\int
_{0}^{\infty}r^{N+2k-3}|w|^{2}dr\right)  \left(  \int_{0}^{\infty}%
r^{N+2k-1}|w|^{2}dr\right)  \nonumber\\
&  \geq\frac{\left(  N+2k+1\right)  ^{2}}{4}\left(  \int_{0}^{\infty
}r^{N+2k-2}|w|^{2}dr\right)  ^{2},\quad\forall k\geq0.
\end{align}
Therefore, by Lemma \ref{l3.4}, we get
\begin{align*}
\nu^{\star\ast}(N) &  \geq\inf_{k\geq0}\frac{\frac{\left(  N+2k+1\right)
^{2}}{4}}{\left(  1+\frac{4k}{(N+2k-3)^{2}}\right)  ^{2}}\\
&  =\inf_{k\geq0}\frac{\left(  N+2k+1\right)  ^{2}}{4}\frac{(N+2k-3)^{4}%
}{\left(  (N+2k-3)^{2}+4k\right)  ^{2}}\\
&  =\frac{(N+1)^{2}}{4}\text{ for }N\geq5.
\end{align*}
Hence, by standard density argument, we have for $u\in W^{2,2}\left(
\mathbb{R}^{N}\right)  $ that
\begin{equation}
\int_{\mathbb{R}^{N}}|\Delta u|^{2}\mathrm{dx}\int_{\mathbb{R}^{N}}|\nabla
u|^{2}\mathrm{dx}\geq\frac{(N+1)^{2}}{4}\left(  \int_{\mathbb{R}^{N}}%
\frac{|\nabla u|^{2}}{\left\vert x\right\vert }\mathrm{dx}\right)  ^{2}.
\end{equation}

\paragraph{\textbf{Attainability of the best constant }$\frac{(N+1)^{2}}{4}$}

Now, we show that \eqref{main4} is attained by the function $u\left(
x\right)  =\alpha\left(  1+\beta\left\vert x\right\vert \right)  e^{-\beta
|x|}$, $\beta>0$, $\alpha\in\mathbb{%
%TCIMACRO{\U{2102} }%
%BeginExpansion
\mathbb{C}
%EndExpansion
}$. Indeed, direct computations yield%
\begin{align*}
&  \int_{\mathbb{R}^{N}}|\Delta u|^{2}\mathrm{dx}\\
&  =\left\vert \mathbb{S}^{N-1}\right\vert \int_{0}^{\infty}\left\vert
u^{\prime\prime}\left(  r\right)  +\frac{N-1}{r}u^{\prime}\left(  r\right)
\right\vert ^{2}r^{N-1}dr\\
&  =\left\vert \alpha\right\vert ^{2}\left\vert \beta\right\vert
^{4}\left\vert \mathbb{S}^{N-1}\right\vert \int_{0}^{\infty}\left\vert \beta
r-N\right\vert ^{2}e^{-2\beta r}r^{N-1}dr\\
&  =\left\vert \alpha\right\vert ^{2}\left\vert \beta\right\vert
^{4}\left\vert \mathbb{S}^{N-1}\right\vert \left[  \left\vert \beta\right\vert
^{2}\int_{0}^{\infty}e^{-2\beta r}r^{N+1}dr-2\beta N\int_{0}^{\infty
}e^{-2\beta r}r^{N}dr+N^{2}\int_{0}^{\infty}e^{-2\beta r}r^{N-1}dr\right]  .
\end{align*}
Note that
\[
\int_{0}^{\infty}e^{-2\beta r}r^{N+1}dr=\frac{1}{\beta^{N+2}}\int_{0}^{\infty
}e^{-2r}r^{N+1}dr=\frac{\left(  N+1\right)  !}{\beta^{N+2}2^{N+2}},
\]%
\[
\int_{0}^{\infty}e^{-2\beta r}r^{N}dr=\frac{1}{\beta^{N+1}}\int_{0}^{\infty
}e^{-2r}r^{N}dr=\frac{N!}{\beta^{N+1}2^{N+1}},
\]
and%
\[
\int_{0}^{\infty}e^{-2\beta r}r^{N-1}dr=\frac{1}{\beta^{N}}\int_{0}^{\infty
}e^{-2r}r^{N-1}dr=\frac{\left(  N-1\right)  !}{\beta^{N}2^{N}}.
\]
Hence%
\[
\int_{\mathbb{R}^{N}}|\Delta u|^{2}\mathrm{dx}=\left\vert \alpha\right\vert
^{2}\left\vert \beta\right\vert ^{4}\left\vert \mathbb{S}^{N-1}\right\vert
\frac{\left(  N+1\right)  !}{\beta^{N}2^{N+2}}.
\]
Also,
\begin{align*}
\int_{\mathbb{R}^{N}}\frac{|\nabla u|^{2}}{\left\vert x\right\vert
}\mathrm{dx} &  =\left\vert \mathbb{S}^{N-1}\right\vert \int_{0}^{\infty
}\left\vert u^{\prime}\left(  r\right)  \right\vert ^{2}r^{N-2}dr\\
&  =\left\vert \alpha\right\vert ^{2}\left\vert \beta\right\vert
^{4}\left\vert \mathbb{S}^{N-1}\right\vert \int_{0}^{\infty}e^{-2\beta r}%
r^{N}dr\\
&  =\left\vert \alpha\right\vert ^{2}\left\vert \beta\right\vert
^{4}\left\vert \mathbb{S}^{N-1}\right\vert \frac{N!}{\beta^{N+1}2^{N+1}}%
\end{align*}
and%
\begin{align*}
\int_{\mathbb{R}^{N}}|\nabla u|^{2}\mathrm{dx} &  =\left\vert \mathbb{S}%
^{N-1}\right\vert \int_{0}^{\infty}\left\vert u^{\prime}\left(  r\right)
\right\vert ^{2}r^{N-1}dr\\
&  =\left\vert \alpha\right\vert ^{2}\left\vert \beta\right\vert
^{4}\left\vert \mathbb{S}^{N-1}\right\vert \int_{0}^{\infty}e^{-2\beta
r}r^{N+1}dr\\
&  =\left\vert \alpha\right\vert ^{2}\left\vert \beta\right\vert
^{4}\left\vert \mathbb{S}^{N-1}\right\vert \frac{\left(  N+1\right)  !}%
{\beta^{N+2}2^{N+2}}%
\end{align*}
Therefore
\[
\frac{\int_{\mathbb{R}^{N}}|\Delta u|^{2}\mathrm{dx}\int_{\mathbb{R}^{N}%
}|\nabla u|^{2}\mathrm{dx}}{\left(  \int_{\mathbb{R}^{N}}\frac{|\nabla u|^{2}%
}{\left\vert x\right\vert }\mathrm{dx}\right)  ^{2}}=\frac{\frac{\left(
N+1\right)  !}{2^{N+2}}\frac{\left(  N+1\right)  !}{2^{N+2}}}{\left(
\frac{N!}{2^{N+1}}\right)  ^{2}}=\frac{(N+1)^{2}}{4}\text{.}%
\]

\subsection{Proof of Theorem \ref{th4}}

Let $u\in C_{c}^{\infty}(\mathbb{R}^{N})$. By Lemma \ref{l3.2}, inequality
\eqref{main5} is equivalent to
\begin{multline}
\sum_{k=0}^{\infty}\left(  \int_{0}^{\infty}\left\vert u_{k}^{\prime\prime
}\left(  r\right)  \right\vert ^{2}r^{N-1}dr+\left(  N-1\right)  \int
_{0}^{\infty}\left\vert u_{k}^{\prime}\left(  r\right)  \right\vert
^{2}r^{N-3}dr\right)  \\
\times\sum_{k=0}^{\infty}\left(  \int_{0}^{\infty}\left\vert u_{k}^{\prime
}\left(  r\right)  \right\vert ^{2}r^{N-1}dr\right)  \\
\geq\frac{(N+1)^{2}}{4}\left(  \sum_{k=0}^{\infty}\int_{0}^{\infty}\left\vert
u_{k}^{\prime}\left(  r\right)  \right\vert ^{2}r^{N-2}dr\right)  ^{2}%
\end{multline}
Obviously, it is enough to show for any $k\geq0$ that
\begin{multline}
\left(  \int_{0}^{\infty}\left\vert u_{k}^{\prime\prime}\left(  r\right)
\right\vert ^{2}r^{N-1}dr+\left(  N-1\right)  \int_{0}^{\infty}\left\vert
u_{k}^{\prime}\left(  r\right)  \right\vert ^{2}r^{N-3}dr\right)  \\
\times\left(  \int_{0}^{\infty}\left\vert u_{k}^{\prime}\left(  r\right)
\right\vert ^{2}r^{N-1}dr\right)  \\
\geq\frac{(N+1)^{2}}{4}\left(  \int_{0}^{\infty}\left\vert u_{k}^{\prime
}\left(  r\right)  \right\vert ^{2}r^{N-2}dr\right)  ^{2}.
\end{multline}
But this is just (\ref{ineq_key_1}) with $k=0$. Also, it is easy to see that
with $u\left(  x\right)  =\alpha\left(  1+\beta\left\vert x\right\vert
\right)  e^{-\beta|x}$, $\beta>0$, $\alpha\in\mathbb{%
%TCIMACRO{\U{2102} }%
%BeginExpansion
\mathbb{C}
%EndExpansion
}$, then
\[
\int_{\mathbb{R}^{N}}|\mathcal{R}_{2}u|^{2}\mathrm{dx}\int_{\mathbb{R}^{N}%
}|\mathcal{R}_{1}u|^{2}\mathrm{dx}=\frac{(N+1)^{2}}{4}\left(  \int
_{\mathbb{R}^{N}}\frac{|\mathcal{R}_{1}u|^{2}}{\left\vert x\right\vert
}\mathrm{dx}\right)  ^{2}.
\]

\end{document}